\documentclass{article}

\usepackage[english]{babel}
\usepackage[dvipsnames,usenames]{color}

\usepackage{stackengine}
\usepackage{scalerel}
 
\stackMath

\newcommand{\nats}                  {{\mathbb N}}

\newcommand{\Mt} {\mathbb{M}}

\newcommand{\botop}{\ensuremath{\bot\mkern-14mu\top}}

\usepackage{verbatim,graphics,indentfirst} 
\usepackage{url}
\usepackage{stmaryrd}
\usepackage[nosumlimits]{amsmath}
\usepackage{amssymb}
\usepackage{amsthm}
\usepackage{latexsym}
\usepackage{enumerate}
\usepackage{graphicx}
\usepackage[T1]{fontenc}
\usepackage{hyperref}
\usepackage{mathrsfs} 
\usepackage{setspace}
\usepackage[all]{xy}

\usepackage{tikz,xparse,xstring}
\usetikzlibrary{calc}

\newcommand{\MyArc}[1][1]{%
    \begin{scope}[#1]
        \fill[white]
            (0,0) arc (270:180:4*\ArW) -- (0,4*\ArW) -- cycle ;
        \draw   (0,0) arc (270:180:4*\ArW) ;
        \fill[white]
            (0,0) arc (90:180:4*\ArW) -- (0,-4*\ArW) -- cycle ;
        \draw (0,0) arc (90:180:4*\ArW) ;
    \end{scope}
}

\newcommand{\MyArcRv}[1][1]{%
    \begin{scope}[#1]
        \fill[white] (0,4*\ArW) arc (0:-90:4*\ArW)
            arc (90:0:4*\ArW) -- cycle ;
        \draw (0,4*\ArW) arc (0:-90:4*\ArW)
            arc (90:0:4*\ArW) ;
    \end{scope}
}

\tikzset{%
    Arrow width/.store in=\ArW,
    Arrow width=.8pt,
    pics/.cd,
    arc east/.style={code = {\MyArc[xshift=.5*\pgflinewidth]}},
    arc west/.style={code = {\MyArc[rotate=180,xshift=.5*\pgflinewidth]}},
    arc rv east/.style={code = {\MyArcRv[xshift=.5*\pgflinewidth]}},
    arc rv west/.style={code = {\MyArcRv[rotate=180,xshift=.5*\pgflinewidth]}}
}

\NewDocumentCommand{\imparrow}{
    D<>{} 
    O{.9em} 
    m 
    D<>{} 
    }{%
    \def\Law{}
    \def\Raw{}
    \def\Style{}
    \IfBeginWith{#3}{<}{\def\Law{pic{arc west}}}{}
    \IfBeginWith{#3}{>}{\def\Law{pic{arc rv west}}}{}
    \IfEndWith{#3}{>}{\def\Raw{pic{arc east}}}{}
    \IfEndWith{#3}{<}{\def\Raw{pic{arc rv east}}}{}
    \IfSubStr[1]{#3}{=}{\def\Style{double}}{}
    \tikz [baseline=-.5ex,line width=.4pt,double distance=1.4pt]
        \draw[\Style] (-3pt,0) (0,0) \Law
            -- node[below=-2pt] {#1} node[above=-2pt] {#4}
            (#2,0) \Raw
            ++(3pt,0);
     }

\NewDocumentCommand{\imparrowlong}{
    D<>{} 
    O{1.1em} 
    m 
    D<>{} 
    }{%
    \def\Law{}
    \def\Raw{}
    \def\Style{}
    \IfBeginWith{#3}{<}{\def\Law{pic{arc west}}}{}
    \IfBeginWith{#3}{>}{\def\Law{pic{arc rv west}}}{}
    \IfEndWith{#3}{>}{\def\Raw{pic{arc east}}}{}
    \IfEndWith{#3}{<}{\def\Raw{pic{arc rv east}}}{}
    \IfSubStr[1]{#3}{=}{\def\Style{double}}{}
    \tikz [baseline=-.5ex,line width=.4pt,double distance=1.4pt]
        \draw[\Style] (-3pt,0) (0,0) \Law
            -- node[below=-2pt] {#1} node[above=-2pt] {#4}
            (#2,0) \Raw
            ++(3pt,0);
     }

\newcommand{\imp}[2]{#1{\!\!{\imparrow{=>}\!\!}}#2}
  \newcommand{\impscript}[2]{#1{\scriptsize\!\!{\imparrow{=>}\!\!}}#2}
\newcommand{\dcoimp}[2]{#2{\!\!{\imparrow{=<}\!\!}}#1}
  \newcommand{\dcoimpscript}[2]{#2{\scriptsize\!\!{\imparrow{=<}\!\!}}#1}
\newcommand{\biimp}[2]{#1{\!\!{\imparrow{<=>}\!\!}}#2}
  \newcommand{\biimpscript}[2]{#1{\scriptsize\!\!{\imparrowlong{<=>}\!\!}}#2}

\usepackage{bbm}

\DeclareMathOperator*{\ifelse}{\textsc{if}}

\newcommand{\conn}{{\copyright}}

\newcommand{\val}{v}
\newcommand{\cB}{\mathcal{B}}
\newcommand{\fibS}{\bullet}
\newcommand{\fib}[2]{#1\fibS#2}
\newcommand{\EqDiv}[3]{#2 \dashv\vdash_{#1} #3}
\newcommand{\TWO}{\mathbbm{2}}

\usepackage{fancyhdr}
\usepackage[yyyymmdd,hhmmss]{datetime}
 \fancypagestyle{plain}{
\fancyhf{}
\lfoot{\thepage}
}
\theoremstyle{definition}
\newtheorem{definition}{\vspace{1mm}Definition}[section]

\theoremstyle{plain}
\newtheorem{lemma}[definition]{\vspace{1mm}Lemma}
\newtheorem{theorem}[definition]{\vspace{1mm}Theorem}
\newtheorem{corollary}[definition]{\vspace{1mm}Corollary}
\newtheorem{proposition}[definition]{\vspace{1mm}Proposition}

\theoremstyle{theorem}
\newtheorem{example}[definition]{\vspace{1mm}Example}
\newtheorem{remark}[definition]{\vspace{1mm}Remark}

\newcommand{\cL}{\mathcal{L}}

\newcommand{\cC}{\mathcal{C}}

\newcommand{\der}    						  {\vdash}

\newcommand{\sH}                              {\mathscr{H}}
\newcommand{\Nat} {\mathbf{N}}
\newcommand{\tuple}[1]                         {{\langle #1\rangle}}

\newcommand{\nmatrix}{Nmatrix}
\newcommand{\nmatrices}{Nmatrices}
\newcommand{\nvalued}{Nvalued}

\DeclareMathOperator*{\sub}{\mathsf{sub}}
\DeclareMathOperator*{\var}{\mathsf{var}}

\DeclareMathOperator*{\head}{\mathsf{head}}
\DeclareMathOperator*{\skel}{\mathsf{skel}}

\newcommand{\ou}          {\vee}
\newcommand{\e}          {\wedge}

\newcommand{\bbtop}{\top\!\!\!\!\top}
\newcommand{\bbbot}{\bot\!\!\!\!\bot}
\newcommand{\bbbotop}{\botop\!\!\!\!\botop}

\newcommand{\bt}{\mathbf{t}}

\newcommand{\Val}{\textrm{Val}}

\newcommand{\ignore}[1]{}

\date{}

\title{Combining fragments of classical logic: \\
When are interaction principles needed?
\footnote{We thank the anonymous referees for their invaluable comments on an earlier version of this paper. This research was done under the scope of R\&D Unit 50008, funded by the applicable financial framework (FCT/MEC, UID/EEA/50008/2013, through national funds and when applicable co-funded by FEDER/PT2020), and is part of the MoSH initiative of SQIG at Instituto de Telecomunicações. 
The third author was partially funded by CNPq.}}

\author{%
Carlos Caleiro \\
  \multicolumn{1}{p{1\textwidth}}{\centering
  \texttt{ccal@math.tecnico.ulisboa.pt}\\
  Dep. Mathematics, IST, Universidade de Lisboa, Portugal, \\
  and SQIG--Instituto de Telecomunica{\c c}{\~o}es}
\and 
S\'ergio Marcelino \\
  \multicolumn{1}{p{1\textwidth}}{\centering
  \texttt{smarcel@math.tecnico.ulisboa.pt}\\
  Dep. Mathematics, IST, Universidade de Lisboa, Portugal, \\ 
  and SQIG--Instituto de Telecomunica{\c c}{\~o}es}
\and
Jo\~ao Marcos \\
  \multicolumn{1}{p{1\textwidth}}{\centering
  \texttt{jmarcos@dimap.ufrn.br}\\
  Lo.L.I.T.A.\ and DIMAp, UFRN, Brazil}
}

\begin{document}

\maketitle

\begin{abstract}
%
We investigate the combination of fragments of classical logic as a way of 
conservatively 
extending a given Boolean logic by the addition of new connectives, and we precisely characterize the circumstances in which such a combination produces the corresponding fragment of classical logic over the signature containing connectives from both fragments given as input.
If the thereby produced combined fragment is only incompletely characterized by the components given as input, this means that connectives from one component need to interact with connectives from the other component, giving rise to interaction principles. 
The main contributions strongly rely on the (well-known) description of the 2-valued clones made by Post, on the (not so well-known) axiomatization procedures for 2-valued matrices laid out by Rautenberg, and on Avron's non-deterministic matrices, which have (recently) been used to produce a significant advance on the understanding of the semantics of fibring. 
%

\end{abstract}

%
%
%
%
%

\section{Introduction}

\noindent 
In what concerns the extensibility of the language of a given logic by some~new connective respecting certain inferential patterns, one of the main criteria invoked in justifying, granting intelligibility, or acknowledging the legitimacy of such an extension is the `conservativeness restraint'.  According to such restraint, the addition of a new connective together with its corresponding characterizing rules should not allow for novel inferences to arise using exclusively the original language, involving formulas deprived of such connective.  
Arguably, an equally important but much less discussed criterion involves the 
 {possible emergence}, through such extension, of `interaction principles' involving the newly added connective and other connectives from the original language extended therewith.

The most common proof formalisms used in the literature in discussing how rules give meaning to the connectives they govern, originated from the landmark work of Gentzen \cite{gen:34}, typically allow for interaction to arise in rather unexpected ways.
For an example, one might recall that logics containing conjunction and disjunction often have as algebraic counterparts some variety of lattices or another.  However, the existence of non-distributive lattices does not seem to be matched in a natural way by logics whose disjunction does not distribute over conjunction.  Quite to the contrary, the canonical presentations of the latter connectives in natural deduction or sequent calculi in general enjoy distributivity as an artifact that is produced by the very choice of proof formalisms (cf.~\cite{jybcon,hum:and-or:2015}). 
Excessive interaction might also be held responsible for `collapsing phenomena' in which two connectives turn out to be indistinguishable when their rules are put together for the definition of a single logic containing both connectives.
There is for instance a well-known debate in the literature about the presentation of a logic containing both a classical and an intuitionistic implication (cf.~\cite{LdCer:AHer:combining:1996}).  The common arguments according to which these two implications would necessarily coincide are however based either on the (incorrect) assumption that the minimal logic that contains two standard implications enjoys an unrestricted version of the Deduction Metatheorem, or on some (incidental or artificial) demand for 
other meta-properties that are expressed in a Gentzen-style formalism
(cf.~\cite{Gabbay:weaving1,cal:ram:07:collapsing,metafib}). 

The main known mechanisms for combining logics often differ on how they deal with conservativeness and interaction.  Among such combination mechanisms, \textit{fibring} fares well on both fronts: unintended interaction is unlikely to arise through fibring, and the fibring of two logics containing no quasi theorems (formulas that follow from whatsoever non-empty set of premises) is always conservative over each component (cf.~\cite{charac-conserv:JLC}).  
Within the scope of such a combination mechanism the ideas concerning the addition of a new connective to a given logic can be made clear and distinct, and the related questions may be given precise answers.
It is worth noting, in particular, that the smallest logic that conservatively extends both the `logic of conjunction' and the `logic of disjunction' is not distributive (cf.~\cite{charac-conserv:JLC}), and also noting that the smallest logic that conservatively extends both the logics of classical implication and of intuitionistic implication does not actually necessitate the collapse between the latter connectives (cf.~\cite{cal:ram:07:clas-int}).  In fact, the results in the present paper imply that it is even plausible to have two non-collapsing copies of classical implication cohabiting the same logic.
In both the above mentioned examples, and in many others, the corresponding joint fragments of classical logic can be recovered by the addition of inference rules capturing the emerging interaction principles.

A neat characterization of fibring is given by way of Hilbert calculi: the combination of two logics, each one characterized by a certain set of inference rules, is produced by the union of these sets of rules.
In contrast ---and in a sense precisely for being so frugal on what concerns interaction principles--- fibring resisted admitting a straightforward semantics (see~\cite{handbook,cutandpaste} for an overview).  
Indeed, among other phenomena to be discussed in the present contribution, it is worth noting that one could very well happen to fiber the logics of two connectives with 2-valued semantics and end up giving origin to a logic with no finite-valued semantics whatsoever, even if non-determinism were allowed.
Nonetheless, after an important theoretical advance contributed by~\cite{smar:ccal:17}, we now know that a semantics for disjoint fibring may be given through a powerful and elegant technology that makes use of non-deterministic semantics.  This technology is applied in the present paper to the combination of fragments of classical logic, as a way of illustrating how rich is the problem that the new semantics allows solving.

The paper is organized as follows. In Section~\ref{prel} we recall a number of necessary definitions and facts regarding logics, their semantics and axiomatizations. In particular, we introduce logical matrices and Nmatrices, as well as some important properties and operations on them.
We put special emphasis on classical logic, and on Post's characterization of Boolean clones. We also recall the essential mechanism of fibring, and we prove some useful results about fibred logics and their derived connectives.  
Several fundamental facts about disjoint fibrings of fragments of classical logic and the characterizations of the resulting logics are then proved in Section~\ref{sec:merging}, along with several illustrative examples.  
The general plan draws to a close, in Section~\ref{sec:summary}, by proving the main announced results concerning the combination of fragments of classical logic and by a recollection of what has been accomplished along the way towards attaining the stated goals.  This is followed in Section~\ref{sec:conclusion} by some pointers to directions for future research.

\section{Preliminaries}\label{prel}

\noindent 
This section contains the main definitions, fixes the notation for the rest of the paper, recalls several important notions and well-known results, makes some remarks, and presents a few new simple and useful facts. %

{
\subsection{Syntax}

A propositional \emph{signature} $\Sigma$ is a family $\{\Sigma^{(k)}\}_{k\in \nats}$ of sets, where each $\Sigma^{(k)}$ contains the \emph{$k$-place connectives} of $\Sigma$. 
To simplify notation, 
we express the fact that $\conn\in\Sigma^{(k)}$ for some $k\in\nats$ by simply writing $\conn\in\Sigma$, and we write $\Sigma_1\cup\Sigma_2$ (resp., $\Sigma_1\cap\Sigma_2$) to denote the signature~$\Sigma$ such that $\Sigma^{(k)}=\Sigma_1^{(k)}\cup\Sigma_2^{(k)}$ (resp., $\Sigma^{(k)}=\Sigma_1^{(k)}\cap\Sigma_2^{(k)}$) for all $k\in \nats$.  We also write $\Sigma_1\subseteq\Sigma_2$ when $\Sigma_1^{(k)}\subseteq\Sigma_2^{(k)}$ for all $k\in \nats$.
The signatures~$\Sigma_1$ and~$\Sigma_2$ are said to be \emph{disjoint} when $\Sigma_1\cap\Sigma_2=\varnothing$.
The language $L_{\Sigma}(P)$ is the carrier of the absolutely free $\Sigma$-algebra generated over a given 
set of sentential variables $P$. Elements of $L_{\Sigma}(P)$ are called \emph{formulas}.  Given a formula $\varphi\in L_{\Sigma}(P)$, we denote by $\var(\varphi)$ (resp.\ $\sub(\varphi)$) the set of variables (resp.\ subformulas) of~$\varphi$, recursively defined as usual; the extension of $\var$ and $\sub$ from formulas to sets thereof is defined as one would expect. We say that two (sets of) formulas \emph{share no variables} if their underlying sets of variables are disjoint. 
If $\varphi\notin P$ we say that~$\varphi$ is \emph{compound},  and we denote by $\head(\varphi)$ its outermost connective.
As usual, given a $1$-place connective~$\conn$, we define the possible \emph{nestings} of~$\conn$ as $\conn^0p:=p$ and $\conn^{i+1}p:=\conn(\conn^i p)$.
When appropriate, given any symbol~$\mathsf{s}$, we will use $\overline{\mathsf{s}}^k$ to denote a sequence of~$k$ consecutive occurrences of $\mathsf{s}$.

A \emph{substitution} is a mapping $\sigma:P\longrightarrow  L_\Sigma(P)$,
uniquely extendable into an endomorphism $\cdot^\sigma:L_\Sigma(P)\longrightarrow  L_\Sigma(P)$.
Given $\Gamma\subseteq L_{\Sigma}(P)$, we denote by~$\Gamma^\sigma$ the set $\{\varphi^\sigma:\varphi\in\Gamma\}$. 
We take a \emph{$k$-place derived connective} $\lambda p_1\dots p_k.\,\varphi$, also denoted by $\varphi(p_1,\dots,p_k)$ when convenient, to be a formula~$\varphi\in L_\Sigma(\{p_1,\ldots,p_k\})$. 
%
Given two signatures~$\Xi$ and~$\Sigma$, a (\emph{homophonic}) \emph{translation} $\bt:\Xi\longrightarrow  L_\Sigma(P)$ is a mapping that assigns to each $k$-place connective $\xi\in\Xi$ a formula $\bt(\xi)\in L_\Sigma(\{p_1,\dots,p_k\})$ (understood as a derived $k$-place connective $\lambda p_1\dots p_k.\,\bt(\xi)$). Such translation extends naturally into a function $\bt:L_\Xi(P)\longrightarrow  L_\Sigma(P)$, defined by setting $\bt(p):=p$ for $p\in P$, and $\bt(\xi(\psi_1,\ldots,\psi_k)):=\bt(\xi)(\bt(\psi_1),\dots,\bt(\psi_k))$ for $\xi\in\Xi^{(k)}$. We use $\mathbf{id}_\Sigma:\Sigma\longrightarrow  L_\Sigma(P)$ to refer to the \emph{identity} translation defined by setting $\mathbf{id}_\Sigma(\conn):=\conn(p_1,\dots,p_k)$ for each $k$-place connective $\conn\in\Sigma$. Given disjoint signatures~$\Xi_1$ and~$\Xi_2$, and translations $\bt_1:\Xi_1\longrightarrow  L_{\Sigma_1}(P)$ and $\bt_2:\Xi_2\longrightarrow  L_{\Sigma_2}(P)$, we use $\bt_1\cup\bt_2:\Xi_1\cup\Xi_2\longrightarrow  L_{\Sigma_1\cup\Sigma_2}(P)$ to denote their union.

Given signatures $\Sigma\subseteq\Xi$, let $X_{\Sigma}:=\{x_\varphi:\varphi\in L_{\Xi}(P)\setminus P\mbox{ and }\head(\varphi)\notin\Sigma\}$ be a new set of sentential variables. Using $X_{\Sigma}$ to see as `monoliths' the formulas from~$\Xi$ whose heads are alien to~$\Sigma$, we can represent in $L_{\Sigma}(P\cup X_{\Sigma})$ the $\Sigma$-\emph{skeleton} of any formula $\varphi\in L_{\Xi}(P)$ by setting $\skel_{\Sigma}(p):=p$ if $p\in P$, and setting for each connective $\conn\in\Xi^{(k)}$:
\smallskip

\noindent
$
\skel_{\Sigma}(\conn(\varphi_1,\dots,\varphi_k)):= 
\begin{cases}
    \conn(\skel_{\Sigma}(\varphi_1),\dots,\skel_{\Sigma}(\varphi_k)),& \text{if }\conn\in\Sigma \\
    x_{\conn(\varphi_1,\dots,\varphi_k)},              & \text{otherwise.}
\end{cases}
$
\smallskip

{ 
\noindent
It is handy to note here that $\sub(\skel_\Sigma(\varphi))\subseteq\skel_\Sigma(\sub(\varphi))$. This implies, given $\Gamma\subseteq L_{\Xi}(P)$, that $\skel_\Sigma(\Gamma)$ is closed under subformulas whenever $\Gamma$ is closed under subformulas. }



\subsection{Logics} 
 
A \emph{logic}~$\cL$ is a structure $\tuple{\Sigma,\der}$, where $\Sigma$ is a {signature} and $\der\ \subseteq 2^{L_{\Sigma}(P)}\times L_{\Sigma}(P)$ is a substitution-invariant (Tarskian) consequence relation over $L_{\Sigma}(P)$. The set $\Gamma\subseteq L_{\Sigma}(P)$ is called an \emph{$\cL$-theory} whenever $\Gamma$ is closed under $\der$, that is $\Gamma^\der:=\{\varphi:\Gamma\der \varphi\}\subseteq \Gamma$. 
We obtain an equivalence relation $\EqDiv{\cL}{}{}$ on sets of formulas of~$\cL$ by defining $\Gamma,\Delta\subseteq L_{\Sigma}(P)$ as (\emph{logically}) \emph{equivalent} when $\Gamma^\der=\Delta^\der$.
An $\cL$-theory~$\Gamma^\der$ is said to be \emph{trivial} if $(\Gamma^\sigma)^\der= L_{\Sigma}(P)$ for every substitution $\sigma:P\longrightarrow  L_\Sigma(P)$, and otherwise said to be \emph{non-trivial}.
Two connectives $\conn_1,\conn_2\in\Sigma^{(k)}$ for some $k\in\nats$ are said to be \emph{indistinguishable} in a logic $\cL=\tuple{\Sigma,\der}$ provided that $\EqDiv{\cL}{\varphi}{\bt(\varphi)}$ for every $\varphi\in L_\Sigma(P)$, where $\bt:\Sigma\to L_\Sigma(P)$ is the translation that replaces 
every occurrence of $\conn_1$ with $\conn_2$, that is, $\bt(\conn_1)=\conn_2(p_1,\dots,p_k)$ and $\bt(\conn)=\conn(p_1,\dots,p_j)$ for every connective $\conn\in\Sigma^{(j)}\setminus\{\conn_1\}$ and every $j\in\nats$.

Let~$\varphi(p_1,\dots,p_k)$ be some $k$-place {derived connective}.  
If $\varphi(p_1,\ldots,p_k)\der p_j$ for some $1\leq j\leq k$, we say that~\emph{$\varphi$ is projective on its $j$-th component}.  Such a derived connective is called a \emph{projection-conjunction} if it is logically equivalent to its set of projective components, i.e., if there is some $J\subseteq\{1,2,\ldots,k\}$ such that (i) $\varphi(p_1,\ldots,p_k)\der p_j$ for every $j\in J$ and (ii) $\{p_j:j\in J\}\der \varphi(p_1,\ldots,p_k)$.  
In case $\varphi(p_1,\ldots,p_k)\der p_{k+1}$, we say that~$\varphi$ is \emph{bottom-like}.
We will call~$\varphi$ \emph{top-like} if $\varnothing\der \varphi(p_1,\ldots,p_k)$. 
Do note that the latter is a particular case of projection-conjunction (take $J=\varnothing$).
Another particular case of projection-conjunction is given by the \emph{affirmation} connective $\lambda p_1.\,p_1$.
A derived connective that is neither top-like nor bottom-like will here be called \emph{significant}; 
if in addition it is not a projection-conjunction, we will call it \emph{very significant}. 
Note that failing to be very significant means being either bottom-like or a projection-conjunction. 
In case $p_1,\dots,p_k\der \varphi(p_1,\ldots,p_k)$, we will say that~$\varphi$ is \emph{truth-preserving}. Obviously, all projection-conjunctions are truth-preserving.

\subsection{Hilbert calculi}

A \emph{Hilbert calculus}~$\sH$ 
is a structure $\tuple{\Sigma,R}$ where $\Sigma$ is a signature, and $R\subseteq 2^{L_\Sigma(P)}\times L_\Sigma(P)$ is a set of so-called \emph{inference rules}. Given $\tuple{\Delta,\psi}\in R$, we refer to $\Delta$ as the set of \emph{premises} and to $\psi$ as the \emph{conclusion} of the rule.  When $\Delta$ is empty, $\psi$ is dubbed an \emph{axiom}. 
		An inference rule $\tuple{\Delta,\psi}\in R$ is often denoted by $\frac{\Delta}{\psi}$, or simply by $\frac{\;\psi_1\;\dots\;\psi_n\;}{\psi}$ if $\Delta=\{\psi_1,\dots,\psi_n\}$ is finite, or by $\frac{}{\;\psi\;}$ if $\Delta=\varnothing$.

	It is well-known that 
	a Hilbert calculus $\sH := \tuple{\Sigma,R}$ induces a logic $\cL_\sH:=\tuple{\Sigma,\der_{\!\sH}}$ such that, for each $\Gamma\subseteq L_{\Sigma}(P)$, $\Gamma^{\,{\der_{\!\sH}}}$ is the least set that contains $\Gamma$ and is closed under all applications of instances of the inference rules in $R$, that is, if $\frac{\Delta}{\psi}\in R$ and $\sigma:P\longrightarrow  L_{\Sigma}(P)$ is such that $\Delta^\sigma\subseteq \Gamma^{\,{\der_{\!\sH}}}$ then $\psi^\sigma\in\Gamma^{\,{\der_{\!\sH}}}$.  
Such definition of a \emph{logic induced by a Hilbert calculus} is meant to capture the `schematic character' of inference rules.

}

\subsection{Logical matrices and \nmatrices{}{}}\label{semantics}

An \emph{\nmatrix{}~$\Mt$ over  a signature~$\Sigma$} is a structure $\tuple{V,D,\cdot_{\Mt}}$ where\footnote{$\tuple{V,\cdot_{\Mt}}$ is a multi-algebra, see~\cite{MR0146103,hyperstructure}.}
$V$ is a
set (of \emph{truth-values}), $D\subseteq V$ is the set of \emph{designated} values and, 
for each $\conn\in \Sigma^{(k)}$, $\cdot_{\Mt}$ gives the interpretation $\conn_{\Mt}:V^k\longrightarrow  2^V\setminus \{\varnothing\}$ of~$\conn$ in $\Mt$.  
%
We use~$U$ to shall refer to the set~$V\setminus D$ of \emph{undesignated} values. 
Henceforth we will assume that we are dealing only with {non-degenerate} \nmatrices{}{}, in the sense that $D\neq\varnothing$ and $U\neq\varnothing$. Clearly, such restriction will only leave out a couple of uninteresting logics. When $D$
is a singleton we will say that~$\Mt$ is \emph{unitary}. 
The traditional, and deterministic, notion of (\emph{logical}) \emph{matrix} is recovered by considering \nmatrices{}{} for which the image of every tuple of values through~$\conn_{\Mt}$ is a singleton, in which case we often drop the braces from the set notation.

A \emph{valuation over~$\Mt$} is a mapping $v:L_{\Sigma}(P)\longrightarrow  V$ such that for each $\conn\in \Sigma^{(k)}$ we have
$v(\conn(\varphi_1,\ldots,\varphi_n))\in \conn_{\Mt}(v(\varphi_1),\ldots,v(\varphi_k))$. We denote by $\Val_P(\Mt)$ the set of all valuations on $L_{\Sigma}(P)$ over $\Mt$.
It is often useful to work with partially defined valuations, i.e., valuations defined only for a certain subset~$\Gamma$ of the language. This is perfectly usual when dealing with logical matrices, as one only needs to define the value of the sentential variables in $\var(\Gamma)$, for then the corresponding valuation extends uniquely to the full language. 
In \nmatrices{}{}, the same effect can be achieved by defining a valuation for a set of formulas~$\Gamma$ that is closed under subformulas, and demanding that it respects the interpretation of connectives, that is, $v(\conn(\varphi_1,\dots,\varphi_n))\in{\conn}_\Mt(v(\varphi_1),\dots,v(\varphi_n))$ for every compound formula $\conn(\varphi_1,\dots,\varphi_n)\in\Gamma$. Such a partial valuation, which we dub a \emph{$\Gamma$-partial valuation}, can always be extended to a valuation over the full language~(cf.~\cite{Avron:MVS-WnH}). 

As usual, we say that a valuation~$v$ over~$\Mt$ \emph{satisfies} a formula~$\varphi$ (resp.\ a set of formulas~$\Gamma$) if $v(\varphi)\in D$ (resp.\ $v(\Gamma)\subseteq D$).
We say that $\Gamma\der_\Mt \varphi$ if every valuation over $\Mt$ that satisfies~$\Gamma$ also satisfies~$\varphi$. It is well known that $\cL_\Mt:=\tuple{\Sigma,\der_\Mt}$ induces a logic, and we call it the \emph{logic characterized by~$\Mt$}.
If~$\Mt$ is a finite \nmatrix{} (i.e., its underlying set of truth-values is finite) then $\cL_\Mt$ is said to be \emph{finitely-\nvalued{}}, or \emph{$k$-\nvalued{}} if $\Mt$ has exactly $k$ truth-values; when~$\Mt$ is a finite logical matrix then~$\cL_\Mt$ is said more simply to be \emph{finitely-valued}, or \emph{$k$-valued}. A~logic $\cL$ is said to be (\emph{deterministically}) \emph{many-valued} if $\cL=\cL_\Mt$ for some logical matrix~$\Mt$ (cf.~\cite{mar:09a:full}).

Given the schematic character of inference rules in Hilbert calculi, we will say that about a valuation~$v$ that it \emph{respects an inference rule $\frac{\Delta}{\psi}$} if, for every substitution $\sigma:P\longrightarrow  L_\Sigma(P)$, we have that $v(\Delta^\sigma)\subseteq D$ implies $v(\psi^\sigma)\in D$.

Consider the signature $\Sigma$ such that $\Sigma^{(k)}=\{\conn\}$ and $\Sigma^{(j)}=\varnothing$ for $j\neq k$.
We will denote by $\bbtop_\conn$ the logic induced, equivalently, by the matrix $\Mt^{\bbtop}_\conn:=\tuple{\{0,1\},\{1\},\cdot_{\bbtop}}$ where $\conn_{\bbtop}(a_1,\dots,a_k)=1$ for all $a_1,\dots,a_k\in\{0,1\}$, or by the Hilbert calculus with the single axiom $\frac{}{\conn(p_1,\dots,p_k)}$, and we will denote by $\bbbot_\conn$ the logic induced, equivalently, by the matrix $\Mt^{\bbbot}_\conn:=\tuple{\{0,1\},\{1\},\cdot_{\bbbot}}$ where $\conn_{\bbbot}(a_1,\dots,a_k)=0$  for all $a_1,\dots,a_k\in\{0,1\}$, or by the Hilbert calculus with the single rule $\frac{\conn(p_1,\dots,p_k)}{p_{k+1}}$. 
It is easy to see that in the former case the $k$-place connective~$\conn$ is a top-like connective, and that in the latter case it is a bottom-like connective.
In addition, by $\bbbotop_\conn$ we will denote the logic of an \emph{unrestrained connective} induced, equivalently, 
by the $2$-valued \nmatrix{} $\Mt^{\bbbotop}_\conn:=\tuple{\{0,1\},\{1\},\cdot_{\bbbotop}}$ where ${\conn}_{\bbbotop}(a_1,\dots,a_k)=\{0,1\}$ for all $a_1,\dots,a_k\in\{0,1\}$, or by the Hilbert calculus with the empty set of rules.

\subsection{Some useful operations on (N)matrices}
\label{sec:useful}

Let $\Xi,\Sigma$ be signatures, $\bt:\Xi\longrightarrow  L_\Sigma(P)$ be a translation, and $\Mt:=\tuple{V,D,\cdot_{\Mt}}$ be a logical matrix over~$\Sigma$.
Then we may say that $\Mt$ \emph{induces an interpretation $\xi_\Mt:V^k\longrightarrow  V$ under~$\bt$} to each connective $\xi\in\Xi$, defined in the case of a $k$-place connective by setting $\xi_\Mt(a_1,\dots,a_k):=v(\bt(\xi))$ where $v$ is any valuation such that $v(p_i)=a_i$ for $1\leq i\leq k$. We denote by $\Mt^\bt$ the matrix over~$\Xi$ with the same truth-values and designated values as $\Mt$, where each $\xi\in\Xi$ receives its interpretation induced under $\bt$. Is is clear that  $\Val_P(\Mt^\bt)=\{v\circ\bt:v\in\Val_P(\Mt)\}$.

Let $\kappa\in\nats\cup\{\omega\}$, with $\kappa>1$. An \nmatrix{} $\Mt:=\tuple{V,D,\cdot_{\Mt}}$ over~$\Sigma$ is said to be \emph{$\kappa$-saturated} if for any sets $\Gamma,\Delta\subseteq L_\Sigma(P)$ with $|\Delta|\leq \kappa$, if $\Gamma\not\der_\Mt\psi$ for each $\psi\in\Delta$ then there exists a valuation~$v$  over $\Mt$ such that $v(\Gamma)\subseteq D$ and $v(\Delta)\subseteq U$. We say that~$\Mt$ is \emph{saturated} if it is $\omega$-saturated (more generally, we might talk about $\kappa$-saturation, where~$\kappa$ the cardinality of the underlying language). 
Note that in a saturated \nmatrix{}~$\Mt$ every $\cL_\Mt$-theory is precisely characterized by a valuation, that is, for every $\cL_\Mt$-theory~$\Gamma$ there is a valuation~$v$ over~$\Mt$ such that $\Gamma=\{\varphi\in L_\Sigma(P):v(\varphi)\in D\}$.
Clearly, if $\Mt$ is $\kappa$-saturated then so is $\Mt^\bt$.

The \emph{$n$-power of~$\Mt$} is the \nmatrix{} $\Mt^n:=\tuple{V^n,D^n,\cdot_n}$ where, for each $k$-place connective $\conn\in\Sigma$, we have $\conn_{n}(\alpha_1,\ldots,\alpha_k)=\{\alpha\in V^n:\pi_i(\alpha)\in \conn_{\Mt}(\pi_i(\alpha_1),\ldots,\pi_i(\alpha_k)) \textrm{ for } 1\leq i\leq n\}$, where each $\pi_i:V^n\longrightarrow  V$ denotes the corresponding $i$-th projection. Note that $\Val_P(\Mt^n)=\Val_P(\Mt)^n$, that is, a valuation on~$\Mt^n$ is just an $n$-tuple of valuations on~$\Mt$.
From~\cite{smar:ccal:17} we know that $\Mt^n$ is $n$-saturated and $\cL_{\Mt}=\cL_{\Mt^n}$, for every \nmatrix{}~$\Mt$.
 Given a translation $\bt:\Xi\longrightarrow  L_\Sigma(P)$, it is straightforward to see that $(\Mt^\bt)^n=(\Mt^n)^\bt$ for every $n\in\nats\cup\{\omega\}$, $n>1$. 

Let $\Sigma_1$ and $\Sigma_2$ be disjoint signatures.
Given \nmatrices{}{} $\Mt_1:=\tuple{V_1,D_1,\cdot_{\Mt_1}}$ over~$\Sigma_1$ and $\Mt_2:=\tuple{V_2,D_2,\cdot_{\Mt_2}}$ over~$\Sigma_2$, 
their \emph{strict product} $\Mt_1\star\Mt_2$ is the \nmatrix{} over $\Sigma_1\cup \Sigma_2$ defined by $\tuple{V_{12},D_{12},\cdot_{\star}}$ where $V_{12}=(D_1\times D_2) \cup (U_1\times U_2)$, $D_{12}=D_1\times D_2$, and for each $k$-place $\conn\in\Sigma_1\cup \Sigma_2$,
$$\conn_{\star}((a_1,b_1),\ldots,(a_k,b_k)):=\begin{cases}
 \{(a,b)\in V_{12}:
a  \in \conn_{\Mt_1}(a_1,\ldots,a_k)\}, \mbox { if }\conn\in \Sigma_1\\[.1cm]
 \{(a,b)\in V_{12}\,:
b \in \conn_{\Mt_2}(b_1\,\ldots,b_k)\}, \mbox { if }\conn\in \Sigma_2
\end{cases}
$$
\noindent
Note that a valuation $v$ over $\Mt_1\star\Mt_2$ 
has two projections $\pi_1(v)$ and $\pi_2(v)$ which (under the obvious restrictions to $L_{\Sigma_1}(P)$ and $L_{\Sigma_2}(P)$) are valuations over~$\Mt_1$ and~$\Mt_2$. 
We know from~\cite{smar:ccal:17} that $\Mt_1\star\Mt_2$ is saturated when both $\Mt_1$ and $\Mt_2$ are saturated.
{

The following lemma
 is very useful in practice, as it tells us how to build in a component-wise manner valuations in an \nmatrix{} obtained by strict product. 
Recall that given a $\Sigma$-\nmatrix{} $\Mt$, if $v$ is a $\Gamma$-partial valuation over $\Mt$ with $\Gamma\subseteq L_\Sigma(P)$, and we are given a sentential variable $p\notin\var(\Gamma)$, then~$v$ may always be extended to a $(\Gamma\cup\{p\})$-partial valuation $v'$ by assigning $v'(p)=a$ for any truth-value~$a$ in the set of truth-values, chosen to be designated, or undesignated, if desired.

\begin{lemma}\label{mergingvals}
Let $\Sigma_1$ and $\Sigma_2$ be disjoint signatures, let ${\Mt}_1$ be a $\Sigma_1$-\nmatrix{} and let ${\Mt}_2$ be a $\Sigma_2$-\nmatrix{}.
Further, let $\Gamma\subseteq L_{\Sigma_1\cup\Sigma_2}(P)$ be closed under subformulas, and take~$v_1$ as a $\skel_{\Sigma_1}(\Gamma)$-partial valuation over ${\Mt}_1$, and~$v_2$ as a $\skel_{\Sigma_2}(\Gamma)$-partial valuation over ${\Mt}_2$.

If the following \emph{compatibility} condition holds:
\begin{itemize}
\item[] $v_1(\skel_{\Sigma_1}(\varphi))\in D_1$ iff $v_2(\skel_{\Sigma_2}(\varphi))\in D_2$ for all $\varphi\in \Gamma$,
\end{itemize}
then setting $v(\varphi)=(v_1(\skel_{\Sigma_1}(\varphi)),v_2(\skel_{\Sigma_2}(\varphi)))$, for $\varphi\in \Gamma$, defines a $\Gamma$-partial valuation over ${\Mt}_1\star {\Mt}_2$. 
\end{lemma}
\begin{proof}
The compatibility condition guarantees that for each $\varphi\in \Gamma$ the pair $(v_1(\skel_{\Sigma_1}(\varphi)),v_2(\skel_{\Sigma_2}(\varphi)))$  is a truth-value of ${\Mt}_1\star {\Mt}_2$. One just needs to check that the interpretation of connectives is respected. Assume, without loss of generality, that $\varphi=\conn(\varphi_1,\dots,\varphi_n)\in \Gamma$ with $\conn\in\Sigma_1$. Since $v_1$ is a $\skel_{\Sigma_1}(\Gamma)$-partial valuation over ${\Mt}_1$ we know that $v_1(\skel_{\Sigma_1}(\varphi))\in\widetilde{\conn}(v_1(\skel_{\Sigma_1}(\varphi_1)),\dots,\linebreak v_1(\skel_{\Sigma_1}(\varphi_n)))$. Therefore, 

$v(\varphi)=(v_1(\skel_{\Sigma_1}(\varphi)),v_2(\skel_{\Sigma_2}(\varphi)))$

$\in\widetilde{\conn}((v_1(\skel_{\Sigma_1}(\varphi_1)),v_2(\skel_{\Sigma_2}(\varphi_1))),\dots,(v_1(\skel_{\Sigma_1}(\varphi_n)),v_2(\skel_{\Sigma_2}(\varphi_n))))$

$=\widetilde{\conn}(v(\varphi_1),\dots,v(\varphi_n))$.
\end{proof}
\noindent 
Hereupon, the $\Gamma$-partial valuation $v$ built as in the proof of the above lemma will be denoted by $v_1\star v_2$.

Take a valuation $v$ over $\Mt_1\star\Mt_2$. If we understand now $\pi_1(v)$ and $\pi_2(v)$ as transformed into functions $\pi_i(v):L_{\Sigma_i}(P\cup X_i)\longrightarrow V_i$ in the obvious way, then it is clear that they are compatible in the above sense, and that $v=\pi_1(v)\star\pi_2(v)$. 
In other words, 
$\Val_P(\Mt_1\star\Mt_2)=\{v_1\star v_2: v_1\in\Val_{P\cup X_{\Sigma_1}}(\Mt_1)\textrm{ is compatible with }\linebreak v_2\in\Val_{P\cup X_{\Sigma_2}}(\Mt_2)\}$.
}

\subsection{Classical logic}\label{classical}

Classical logic, in any desired signature $\Sigma$, is $2$-valued. We shall denote by $\TWO_\Sigma$ the matrix $\tuple{\{0,1\},\{1\},\cdot_\TWO}$ where $\conn_\TWO=\widetilde{\conn}:\{0,1\}^k\longrightarrow \{0,1\}$ is the Boolean function associated to each $k$-place Boolean connective $\conn\in\Sigma$.

The most common Boolean connectives, namely $\top$ and $\bot$ ($0$-place), $\neg$ ($1$-place), $\e,\ou$ and $\imp{}{}$ ($2$-place) have their interpretations given through the following tables.
%
\begin{center}
    \begin{minipage}{0.18\textwidth}
        \centering
          \begin{tabular}{c}
    $\widetilde{\top}$  \\ 
    \hline
    $1$ \\
    \mbox{}
  \end{tabular}
  \quad
          \begin{tabular}{c}
    $\widetilde{\bot}$   \\ 
    \hline
    $0$\\
    \mbox{}
  \end{tabular}
    \end{minipage}%
    \begin{minipage}{0.18\textwidth}
        \centering
          \begin{tabular}{c | c c}
    &  $\widetilde{\neg}$   \\ 
    \hline
    $0$ & $1$  \\ 
     $1$  &$0$  
  \end{tabular}
    \end{minipage}    
    \begin{minipage}{0.2\textwidth}
        \centering
   \begin{tabular}{c | c c}
    $\widetilde{\e}$&  $0$ & $1$   \\ 
    \hline
    $0$ & $0$ &  $0$ \\ 
     $1$  &$0$  &  $1$  
  \end{tabular}     
    \end{minipage}
    \begin{minipage}{0.2\textwidth}
        \centering
    \begin{tabular}{c | c c}
    $\widetilde{\ou}$&  $0$ & $1$   \\ 
    \hline
    $0$ & $0$ &  $1$ \\ 
     $1$  &$1$  &  $1$  
  \end{tabular}        
    \end{minipage}
    \begin{minipage}{0.2\textwidth}
        \centering
    \begin{tabular}{c | c c}
    $\widetilde{\smash{\imp{}{}}}$&  $0$ & $1$   \\ 
    \hline
    $0$ & $1$ &  $1$ \\ 
     $1$  &$0$  &  $1$  
  \end{tabular}        
    \end{minipage}
\end{center}
%
Valuations over $\TWO_\Sigma$ are dubbed \emph{bivaluations}. We use $\cB_\Sigma=\cL_{\TWO_\Sigma}$ to denote the \emph{$\Sigma$-fragment} of classical logic, and use $\der_{\cB_\Sigma}$ to denote the associated consequence relation.

Hilbert calculi for the corresponding one-connective fragments of classical logic are well known, or may be systematically obtained from sections~2 and~3 of~\cite{Rautenberg1981}. Possible axiomatizations for the above mentioned connectives are listed below:
\smallskip

\noindent
\scalebox{.95}{
\hspace{-3mm}
\begin{tabular}{ll}
  {[$\cB_\top$]} & $\frac{}{\;\;\top\;\;}\ignore{{}^{{}_{{}_{\mathsf{t}1}}}}$\\[2mm]
  {[$\cB_\bot$]} & $\frac{\;\;\bot\;\;}{p}\ignore{{}^{{}_{{}_{\mathsf{b}1}}}}$\\[2mm]
  {[$\cB_\neg$]} & 
    $\frac{p}{\;\;\neg\neg p\;\;}\ignore{{}^{{}_{{}_{\mathsf{n}1}}}}\quad
    \frac{\;\;\neg \neg p\;\;}{p}\ignore{{}^{{}_{{}_{\mathsf{n}2}}}}\quad 
    \frac{\;\;p\quad  \neg p\;\;}{q}\ignore{{}^{{}_{{}_{\mathsf{n}3}}}}$\\[2mm]
  {[$\cB_\land$]} & 
    $\frac{\;p\e q\;}{\;\; p\;\;}\ignore{{}^{{}_{{}_{\mathsf{c}1}}}}\quad
    \frac{\;p\e q\;}{\;\; q\;\;}\ignore{{}^{{}_{{}_{\mathsf{c}2}}}}\quad
    \frac{\;\;p\quad  q\;\;}{p\e q}\ignore{{}^{{}_{{}_{\mathsf{c}3}}}}$\\[2mm]
  {[$\cB_\lor$]} & 
    $\frac{p}{\;p\ou q\;}\ignore{{}^{{}_{{}_{\mathsf{d}1}}}} \quad
    \frac{\;p\ou p\;}{p}\ignore{{}^{{}_{{}_{\mathsf{d}2}}}} \quad 
    \frac{\;p\ou q\;}{q\ou p}\ignore{{}^{{}_{{}_{\mathsf{d}3}}}} \quad
    \frac{\;p\ou (q \ou r)\;}{(p \ou q)\ou r}\ignore{{}^{{}_{{}_{\mathsf{d}4}}}}$\\[2mm]
  {[$\cB_{\impscript{}{}}$]} & 
  $\frac{}{\;\impscript{p}{(\impscript{q}{p})}\;}\ignore{{}^{{}_{{}_{\mathsf{i}1}}}}\quad
  \frac{}{\;(\impscript{p}{\impscript{(\impscript{q}{r}))}{(\impscript{(\impscript{p}{q})}}{{(\impscript{p}{r}))}}}\;}\ignore{{}^{{}_{{}_{\mathsf{i}2}}}}\quad
  \frac{}{\;\impscript{(\impscript{(\impscript{p}{q})}{p})}{p}\;}\ignore{{}^{{}_{{}_{\mathsf{i}3}}}}\quad 
  \frac{\;p\quad \impscript{p}{q}\;}{q}\ignore{{}^{{}_{{}_{\mathsf{i}4}}}}$\\[2mm]
 \end{tabular}}
 \smallskip

 %
%

Other useful classical connectives may be derived from these, e.g., via a translation~$\bt$ as below:
\smallskip

\noindent
\scalebox{.95}{
\hspace{-3mm}
\begin{tabular}{lll}
$\bt(\dcoimp{}{})$ & $:=$ & $\lambda p_1p_2.\,\neg(\imp{p_2}{p_1})$\\
$\bt(\biimp{}{})$ & $:=$ \hspace{-2mm} & $\lambda p_1p_2.\,(\imp{p_1}{p_2})\e(\imp{p_2}{p_1})$\\
$\bt(+)$ & $:=$ & $\lambda p_1p_2.\,\neg(\biimp{p_1}{p_2})$
  \\
$\bt(\mathsf{+^3})$ & $:=$ & $\lambda p_1p_2p_3.\,p_1+(p_2+p_3)$ \\
\end{tabular}
}\\
\scalebox{.95}{
\hspace{-3mm}
\begin{tabular}{lll}
$\bt(\ifelse)$ & $:=$ & $\lambda p_1p_2p_3.\,(\imp{p_1}{p_2})\land(\neg \imp{p_1}{p_3})$\\
$\bt(T^k_0)$ & $:=$& $\lambda p_1\dots p_k.\,\top$, for $k\geq 0$\\ 
\end{tabular}
}\\
\scalebox{.95}{
\hspace{-3mm}
\begin{tabular}{lll}
$\bt(T^k_k)$ & $:=$& $\lambda p_1\dots p_k.\,p_1\land \dots\land p_k$, for $k>0$\\ 
$\bt(T^k_n)$ & $:=$& $\lambda p_1\dots p_k.\,(p_1\land T^{k-1}_{n-1}(p_2,\dots,p_k))\lor T^{k-1}_{n}(p_2,\dots,p_k)$, for $0<n<k$\\ 
\end{tabular}
}
\smallskip

The Boolean interpretation induced under $\bt(\conn)\in L_\Sigma(P)$ can be immediately obtained from the interpretation of the Boolean connectives in~$\Sigma$ as explained in Subsection~\ref{semantics}, namely $\TWO_\conn:=\TWO^\bt_\Sigma$. Of course, such connectives may be taken as primitive in some fragments of classical logic. The purpose here is just to introduce a general mechanism to produce their interpretations.
Note that $T^k_n$, with $0\leq n\leq k$, represents the so-called $k$-place \emph{threshold} connective such that $\widetilde{T^k_n}(a_1,\dots,a_k)=1$ precisely when 
$n\leq |\{i\in\{1,\dots,k\}:a_i=1\}|$. Axiomatizations for all the corresponding one-connective fragments, or in general for fragments with several connectives, are not always straightforward but may be systematically obtained using the techniques from~\cite{Rautenberg1981}.

{
Given a signature~$\Sigma$ of Boolean connectives, we say that a logic  $\cL=\tuple{\Sigma,\der}$ is \emph{subclassical} whenever 
$\der{\subsetneq} \der_{\TWO_\Sigma}$.}

\begin{remark}\label{rem-significant}
\em
Clearly, $\top$ is a top-like connective, though not all top-like connectives ought to be $0$-place. In the classical setting, a $k$-place connective $\conn$ is top-like precisely in case $\widetilde{\conn}(a_1,\dots,a_k)=1$ for all $a_1,\dots,a_k\in\{0,1\}$, i.e., $\widetilde{\conn}=\widetilde{T^k_0}$. It follows that $\cB_{T^k_0}=\bbtop_{T^k_0}$ for all $k\in\nats$. 
Analogously, $\bot$ is a bottom-like connective, but again not all bottom-like connectives ought to be $0$-place. In the classical setting, a $k$-place connective $\conn$ is bottom-like precisely in case $\widetilde{\conn}(a_1,\dots,a_k)=0$ for all $a_1,\dots,a_k\in\{0,1\}$. It follows that $\cB_\conn=\bbbot_\conn$ when $\conn$ is bottom-like. Apart from $\bot$ and from the projection-conjunctions $\top$, $\e$ and $T^k_k$ for $k\in\nats$, all other Boolean connectives listed above are very significant.
\hfill$\triangle$
\end{remark}

\begin{remark}\label{fact:very-significant}\label{note-boolean}
\em
Classical negation $\neg$ is the only very significant $1$-place Boolean connective. There is only one other significant $1$-place Boolean connective, the affirmation connective, interpreted by setting $\widetilde{\lambda p_1.\,p_1}(a)=a$ for $a\in\{0,1\}$, but it is of course a projection-conjunction. Further, if $\conn$ is any $k$-place very significant Boolean connective and $J\subseteq\{1,\dots,k\}$ is the set of indices of its projective components, then $|J|<k$. In that case, of course, $\conn(p_1,\dots,p_k)\not\der_{\cB_\conn}\conn(p_1,\dots,p_k)^\sigma$ where $\sigma(p_i)=p_i$ if $i\in J$, and 
$\sigma(p_i)=q_i$ if $i\notin J$. Note also that any truth-preserving $k$-place Boolean connective~$\conn$ is such that  $\widetilde{\conn}(\overline{1}^k)=1$. 
\hfill$\triangle$
\end{remark}

Next we state and prove a simple yet quite useful result:

\begin{lemma}\label{lem:nontop}
The logic of a non-top-like $k$-place Boolean connective~$\conn$ with $k>0$ expresses some $1$-place non-top-like compound derived connective~$\theta$.
Furthermore, all possible nestings of~$\theta$ are distinct and none is top-like.
\end{lemma}

\begin{proof}
Let $\alpha$ denote the 1-place derived connective induced by the formula~$\conn(\overline{p}^k)$.
If~$\alpha$ is not top-like, we are done with $\theta=\alpha$.
Otherwise, given that~$
\conn$ is not top-like, there must be some bivaluation~$v$ such that $\val(\conn(p_1,\ldots,p_k))=0$.  Set $I:=\{i:\val(p_i)=1\}$, and define the substitution~$\sigma$ by setting 
$\sigma(p_i):=\alpha(p)$ if $i\in I$, and $\sigma(p_i):=p$ otherwise.
Let $\beta$ denote the new 1-place derived connective induced by $(\conn(p_1,\ldots,p_k))^\sigma$.  Choosing a bivaluation~$\val^\prime$ such that $\val^\prime(p)=0$ we immediately conclude that~$\val^\prime(\beta(p))=\val(\conn(p_1,\ldots,p_k))=0$, and thus $\theta=\beta$ is not top-like.

As~$\theta$ is compound we obtain that $\theta^n(p)\neq\theta^m(p)$ for $n\neq m$.
Clearly, $\theta^0(p)=p$ is not top-like. When $n>0$, if $\theta$ is bottom-like then $\theta^n(p)$ is always bottom-like, if $\theta$ defines affirmation then each $\theta^n(p)$ is also an affirmation connective, and if $\theta$ defines negation then $\theta^n(p)$ alternates between affirmation and negation. In all these cases, it is clear that 
$\not\der_{\cB_\conn}\theta^n(p)$. 
%
\end{proof}

To illustrate the construction in the proof of the above result, consider first Boolean disjunction. The connective~$\lor$ is not top-like, and $\alpha(p):=p\lor p$ is also not.
Consider now Boolean implication. The connective $\imp{}{}$ is also not top-like. However, $\alpha(p):=\imp{p}{p}$ is top-like. Still, $\beta(p):=\imp{(\imp{p}{p})}{p}$ is not top-like.

We shall call $\mathcal{\cC}_2^\Sigma$ the collection of all non-$0$-place Boolean functions compositionally derived (i.e., closed under compositions and projections) over~$\Sigma$, as interpreted through $\TWO_\Sigma$. In the literature on Universal Algebra~\cite{BS81}, $\mathcal{\cC}_2^\Sigma$ is known as the \emph{clone} of operations definable by all derived connectives allowed by the signature~$\Sigma$. We denote simply by $\cC_{2}$ the clone of all non-$0$-place Boolean functions. A~set~$\Sigma$ of Boolean connectives is said to be \emph{functionally complete} precisely when $\mathcal{\cC}_2^\Sigma=\mathcal{\cC}_2$. 

\begin{remark}\label{rem:clones}
\em
Emil Post's characterization of functional completeness for classical logic~\cite{Post41,Lau:2006:FAF:1205006} is very informative. First, it tells us that there are exactly five maximal functionally incomplete clones (i.e., coatoms in Post's lattice), namely $\mathcal{P}_0:=\cC_{2}^{\ou\,\dcoimpscript{}{}}$, $\mathcal{P}_1:=\cC_{2}^{\land\impscript{}{}}$, $\mathcal{A}:=\cC_{2}^{\biimpscript{}{}\bot}$, $\mathcal{M}:=\cC_{2}^{\land\lor\top\bot}$, and $\mathcal{D}:=\cC_{2}^{T^3_2 \neg}$. 

The obvious projection functions $\lambda p_1\ldots p_k.\,p_n$, for $1\leq n\leq k$ and $k\in\Nat$, form the minimal clone $\cC_2^\varnothing$, contained in all the others.
The Boolean top-like connectives form the clone $\mathcal{UP}_1:=\cC_{2}^{\top}$. An analysis of Post's lattice also reveals that there are a number of clones which are maximal with respect to~$\top$, i.e., functionally incomplete clones that become functionally complete by the mere addition of $\top$ (or actually any other connective from $\mathcal{UP}_1$). 
In terms of Post's lattice, the clones whose join with $\mathcal{UP}_1$ result in $\cC_{2}$ are $\mathcal{D}$, $\mathcal{T}^\infty_0:=
\cC_{2}^{\dcoimpscript{}{}}$, and $\mathcal{T}^{n+1}_0:=\cC_{2}^{T^{n+2}_{n+1}\dcoimpscript{}{}}$ for $n\in\nats$. It is worth noting that $\mathcal{T}^1_0=\mathcal{P}_0$. 

Further detailed analysis of Post's lattice also tells us that every clone $\cC_2^\Sigma$ that contains the Boolean interpretation of some very significant connective (i.e., such that $\cC_{2}^\Sigma\not\subseteq \cC_{2}^{\land\top\bot}$) must contain the Boolean function associated  to at least one of the connectives of the following list [\texttt{L0}]: $T^{n+2}_{n+1}$ (for $n\in\nats$), $T^{n+4}_2$ (for $n\in\nats$), $\neg$, $\imp{}{}$, $\biimp{}{}$, $\dcoimp{}{}$, $+$, $\mathsf{+^3}$, $\textsc{if}$,  $\lambda p_1p_2p_3.\,p_1\lor(p_2\land p_3)$, $\lambda p_1p_2p_3.\,p_1\lor(p_2+p_3)$, $\lambda p_1p_2p_3.\,p_1\land(p_2\lor p_3)$, $\lambda p_1p_2p_3.\,p_1\land(\imp{p_2}{p_3})$. \hfill$\triangle$
\end{remark}

What follows is an alternative characterization of very significant Boolean connectives:

\begin{proposition}\label{satveysig}
Let $\Sigma$ be a signature. The matrix $\TWO_\Sigma$ is saturated if and only if $\cC_2^{\Sigma}$ contains no very significant connective.
\end{proposition}
\begin{proof}
Let $\der$ denote $\der_{\cB_\Sigma}$. 
Clearly, $\TWO_\Sigma$ is saturated whenever~$\Sigma$ contains no very significant connective. Indeed, it is straightforward to show by induction on the structure of formulas that, because no connective in~$\Sigma$ is very significant, a non-trivial theory $\Gamma^\der$ is always precisely characterized by a bivaluation~$v$ such that $v(p)=1$ if $\Gamma\der p$, and $v(p)=0$ if $\Gamma\not\der p$, for every $p\in P$.

Now, suppose that $\conn\in\Sigma$ is a $k$-place very significant connective with $j< k$ projective components. We assume without loss of generality that the indices of the projective components of $\conn$ are the first ones. Let $s=k-j$. 
Given the present assumptions, and in view of Rem.~\ref{note-boolean}, given distinct sentential variables $p_1,\dots,p_j,q_1,\dots,q_s,r_1,\dots,r_s\in P$, we have:
\begin{itemize}\setlength\itemsep{0pt}
\item[(a)] $\conn(p_1,\dots,p_j,q_1,\dots,q_s)\der p_i$ for $1\leq i\leq j$
\item[(b)] $\conn(p_1,\dots,p_j,q_1,\dots,q_s)\not\der q_i$ for $1\leq i\leq s$
\item[(c)] $\conn(p_1,\dots,p_j,q_1,\dots,q_s)\not\der\conn(p_1,\dots,p_j,r_1,\dots,r_s)$
\item[(d)] $\conn(p_1,\dots,p_j,q_1,\dots,q_s)\not\der r_i$ for $1\leq i\leq s$
\end{itemize}
If $\TWO_\Sigma$ were saturated then, from (a)--(d), and taking into account the theory $\{\conn(p_1,\dots,p_j,q_1,\dots,q_s)\}^{\der}$, there would exist a bivaluation~$v$ over $\TWO_{\Sigma}$ according to which $v(\conn(p_1,\dots,p_j,q_1,\dots,q_s))=v(p_i)=1$ for 
$1\leq i\leq j$, and simultaneously $v(\conn(p_1,\dots,p_j,r_1,\dots,r_s))=v(q_i)=v(r_i)=0$ for $1\leq i\leq s$. But then
$1=v(\conn(p_1,\dots,p_j,q_1,\dots,q_s))=\widetilde{\conn}(v(p_1),\dots,v(p_j),v(q_1),\dots,v(q_s))=\widetilde{\conn}(\overline{1}^j,\overline{0}^s)=\widetilde{\conn}(v(p_1),\dots,v(p_j),v(r_1),\dots,v(r_s))=v(\conn(p_1,\dots,p_j,r_1,\dots,r_s))\linebreak=0$, which is a contradiction.
%
%
%
\end{proof}

\subsection{Cancellation, tabularity, determinedness}\label{tabularity}

Let $\cL:=\tuple{\Sigma,\der}$ be a logic.
We say that $\cL$ enjoys the \emph{cancellation property} if
$
\Gamma\cup(\bigcup_{i \in I }
 \Delta_i ) \der \varphi$ implies  that $\Gamma \der \varphi
$
	for all $\bigcup_{i\in I} 
\Delta_i \cup \Gamma\cup\{\varphi\}\subseteq L_{\Sigma}(P) $ such that the following conditions hold: 
(i) $\Gamma \cup\{\varphi\}$ shares no variables with $\bigcup_{i\in I} \Delta_i$, 
(ii) $\Delta_i$ shares no variables with~$\Delta_j$, for every $i\neq j\in I$, and 
(iii) $\Delta_i^{\der}$ is non-trivial for every $i\in I$.
It is easy to check that any logic defined by a logical matrix (for instance, classical logic) enjoys the cancellation property. A very interesting result from~\cite{shoesmith1971, Wojcicki1974} 
shows that this property is also a necessary condition for many-valuedness: a logic $\cL$ enjoys cancellation if and only if $\cL=\cL_\Mt$ for some matrix~$\Mt$.

The logic $\cL$ is called \emph{locally tabular} if its associated relation of logical equivalence $\EqDiv{\cL}{}{}$ partitions the language $L_\Sigma(\{p_1,\ldots,p_k\})$, freely generated by the signature~$\Sigma$ over a finite set of sentential variables, into a finite number of equivalence classes. 
It is clear that every logic~$\cB_\Sigma$ is locally tabular --- that constitutes in fact the theoretical underpinning of the classical truth-tabular decision procedure. 
In addition, it is known (for a discussion on this topic see~\cite{finval}) that a logic that fails to be locally tabular cannot be finitely-valued. Do note, however, that a logic may well fail to be locally tabular and yet be finitely-\nvalued{}.

Let $k\in\nats$. The logic~$\cL$ is said to be \emph{$k$-determined} if,
 for all $\Gamma\cup \{\varphi\} \subseteq L_\Sigma(P)$, whenever
 $\Gamma\not\der \varphi$ there is a substitution
 $\sigma \colon P \longrightarrow  \{p_1,\ldots,p_k\}$ such that $\Gamma^\sigma\not\der \varphi^\sigma$. 
 It follows from
 ~\cite{finval} 
 that any $k$-\nvalued{} logic must be $k$-determined, and consequently that if  $k$-determinedness fails for all $k\in\nats$, for a given logic, then this logic cannot be finitely-\nvalued{}.
 
%
%
%

\subsection{Fibred logics}\label{fibring}


	Let $\cL_1:=\tuple{\Sigma_1,\der_1}$ and $\cL_2:=\tuple{\Sigma_2,\der_2}$ be two logics.
	The \emph{fibring} of $\cL_1$ and~$\cL_2$ is the smallest logic $\cL_1\bullet\cL_2:=\tuple{\Sigma_{12},\der_{12}}$ with $\Sigma_{12}=\Sigma_1\cup\Sigma_2$ that extends both $\cL_1$ and $\cL_2$, i.e., such that $\der_1\cup\der_2{\subseteq}\der_{12}$. When the underlying signatures are disjoint, the fibring is said to be \emph{disjoint}. All the phenomena we study in the present paper are instances of disjoint fibring.	Note that, by definition, fibring is commutative and associative, that is, $\cL_1\bullet\cL_2=\cL_2\bullet\cL_1$ and 
	$\cL_1\bullet(\cL_2\bullet\cL_3)=(\cL_1\bullet\cL_2)\bullet\cL_3$ for any given logic $\cL_3$.
	
Given connectives $\conn_1\in\Sigma_1^{(k)}$ and $\conn_2\in\Sigma_2^{(k)}$ for some $k\in\nats$, 
in case $\conn_1$ and $\conn_2$ happen to be indistinguishable in $\cL_1\bullet\cL_2$  
we shall say that $\conn_1,\conn_2$ are \emph{collapsed by fibring $\cL_1$ and $\cL_2$}.
	
	Given Hilbert calculi $\sH_1:=\tuple{\Sigma_1,R_1}$ and $\sH_2:=\tuple{\Sigma_2,R_2}$ 
	then $\cL_1\bullet \cL_2=\cL_{\sH_1\bullet\sH_2}$, where $\sH_{1}\bullet\sH_2:=\tuple{\Sigma_{12},R_1 \cup R_2}$.
	Clearly, besides joining the given signatures, which allows for the construction of so-called `mixed formulas', the fibring of the calculi also allows `mixed reasoning', where rules coming from one logic are used in dealing with formulas coming from the other logic. 
 	 
 
{
The next lemma deals with the semantics of the logic
 obtained by requiring new inference rules to hold in the logic induced by a given \nmatrix{}.
 The first part highlights the role of the notion of saturation, as whenever $R$ contains a non-axiomatic rule then the saturation proviso is fundamental 
 (for an illustration of that, check
Ex.~\ref{negbot}). 

\begin{lemma}\label{addaxioms}\label{addrules}
Let $\Mt$ be an \nmatrix{} over~$\Sigma$ and $\sH:=\tuple{\Sigma,R}$ be a Hilbert calculus. Sufficient conditions for the logic $\cL=\fib{\cL_\Mt}{\cL_\sH}$ to be characterized by $\{v\in\Val_P(\Mt):v\textrm{ respects each }\frac{\,\Delta\,}{\psi}\in R\}$ are secured when \underline{either}:
\begin{itemize}
\item[$\mathbf{(a)}$] $\Mt$ is saturated, \underline{or}
\item[$\mathbf{(b)}$] $R$ contains only axioms.
\end{itemize}
\end{lemma}
\begin{proof}
Both cases are fairly simple. Let $\cL:=\tuple{\Sigma,\der}$.
\begin{itemize}
\item[$\mathbf{(a)}$] As $\der_\Mt{\subseteq}\der$, every $\cL$-theory $\Gamma$ is also an $\cL_\Mt$-theory. Thus, since $\Mt$ is saturated, there exists $v\in\Val_P(\Mt)$ such that $T_v:=\Gamma=\{\varphi:v(\varphi)\in D\}$. Of course, given that $\Gamma$ is an $\cL$-theory it follows that $v$ respects the rules in $R$. Conversely, just observe that $T_v$ is always an $\cL_\Mt$-theory when $v\in\Val_P(\Mt)$, but $T_v$ is also an $\cL$-theory when $v$ respects the rules in $R$.

\item[$\mathbf{(b)}$] Let $\textsl{Ax}=\{\psi^\sigma:\frac{}{\,\psi\,}\in R\text{ and }\sigma:P\longrightarrow L_\Sigma(P)\}$. Observe that $\Gamma\der\varphi$ if and only if $\Gamma\cup\textsl{Ax}\der_\Mt\varphi$. The result follows simply by noting that $v$ respects the axioms in $R$ if and only if $v(\textsl{Ax})\subseteq D$.
\qedhere
\end{itemize}
\end{proof}

}

\begin{remark}\label{sem-fib}
\em
      A semantics for disjoint fibring may be provided through a combo of the operations for strict product and saturation. 
      Assuming $\Sigma_1$ and $\Sigma_2$ to be disjoint, and given \nmatrices{}{}~$\Mt_1$ over~$\Sigma_1$ and~$\Mt_2$ over~$\Sigma_2$,  we know from~\cite{smar:ccal:17} that $\cL_{\Mt_1}\bullet \cL_{\Mt_2}=\cL_{\Mt_1^\omega\star\,\Mt^\omega_2}$. 
      Furthermore, as~$\Mt_i$ is known to be saturated, one can directly use~$\Mt_i$ rather than~$\Mt_i^\omega$, in the latter recipe. \hfill$\triangle$
\end{remark}

Let $\cL:=\tuple{\Sigma,\der}$ be a logic, and $\conn\notin\Sigma$ be any $k$-place connective. The logic resulting from adding~$\conn$ to~$\cL$ as a new unrestrained (resp., top-like / bottom-like) connective is simply $\cL\bullet\bbbotop_\conn$ (resp., $\cL\bullet\bbtop_\conn$ / $\cL\bullet\bbbot_\conn$).

\begin{proposition}\label{withtop}
Given an \nmatrix{} $\Mt:=\tuple{V,D,\cdot_\ast}$ over~$\Sigma$ and a $k$-place $\conn\notin \Sigma$:
\begin{itemize}
\item[$\mathbf{(a)}$] $\cL_{\Mt}\bullet \bbbotop_{\conn}$ is characterized by the \nmatrix{} $\Mt\star\Mt^{\bbbotop}_\conn$ isomorphic to the extension of~$\Mt$ with $\conn_{\star}(a_1,\dots,a_k)=V$ for all $a_1,\dots,a_k\in V$;
\item[$\mathbf{(b)}$] $\cL_{\Mt}\bullet \bbtop_{\conn}$ is characterized by the \nmatrix{} $\Mt\star\Mt^{\bbtop}_\conn$ isomorphic to the extension of~$\Mt$ with $\conn_{\star}(a_1,\dots,a_k)=D$ for all $a_1,\dots,a_k\in V$;
\item[$\mathbf{(c)}$] $\cL_{\Mt}\bullet \bbbot_{\conn}$ is characterized by the \nmatrix{} $\Mt\star\Mt^{\bbbot}_\conn$ isomorphic to the extension of  $\Mt$ with $\conn_{\star}(a_1,\dots,a_k)=U=V\setminus D$ for all $a_1,\dots,a_k\in V$, provided that~$\Mt$ is saturated, or simply $2$-saturated if $k=0$.
\end{itemize}
\end{proposition}
\begin{proof}
First note that $\Mt^{\bbbotop}_\conn$ is saturated. 

Let $\Sigma^+:=\Sigma\cup\{\conn\}$ and fix $\Gamma\cup\{\varphi\}\subseteq L_{\Sigma^+}(P)$.
\begin{itemize}

\item[$\mathbf{(a)}$] Let $\cL_{\Mt}\bullet \bbbotop_{\conn}:=\tuple{\Sigma^+,\der_{\bbbotop}}$. It is easy to see that $\Gamma\der_{\bbbotop}\varphi$ if and only $\skel_\Sigma(\Gamma)\der_\Mt\skel_\Sigma(\varphi)$. Soundness and completeness follow by observing that $\Val_P(\Mt\star\Mt^{\bbbotop}_\conn)=\{v\circ\skel_\Sigma:v\in\Val_{P\cup X_\Sigma}(\Mt)\}$.

\item[$\mathbf{(b)}$] Let $\cL_{\Mt}\bullet \bbtop_{\conn}:=\tuple{\Sigma^+,\der_{\bbtop}}$. It is easy to see that $\Gamma\der_{\bbtop}\varphi$ if and only $\Gamma\cup\{\psi\in L_{\Sigma^+}(P):\head(\psi)=\conn\}\der_{\bbbotop}\varphi$. Soundness and completeness follow by observing that $\Val_P(\Mt\star\Mt^{\bbtop}_\conn)=\{v\in\Val_{P}(\Mt\star\Mt^{\bbbotop}_\conn):v(\psi)\in D \textrm{ for all } \psi\in L_{\Sigma^+}(P)\textrm{ with }\head(\psi)=\conn\}$.

\item[$\mathbf{(c)}$] Let $\cL_{\Mt}\bullet \bbbot_{\conn}:=\tuple{\Sigma^+,\der_{\bbbot}}$. It is easy to see that $\Gamma\der_{\bbbot}\varphi$ if and only $\Gamma\der_{\bbbotop}\varphi$ or $\Gamma\der_{\bbbotop}\psi$ for some $\psi\in L_{\Sigma^+}(P)$ with $\head(\psi)=\conn$. Soundness follows by observing that $\Val_P(\Mt\star\Mt^{\bbbot}_\conn)=\{v\in\Val_{P}(\Mt\star\Mt^{\bbbotop}_\conn):v(\psi)\in U \textrm{ for all } \psi\in L_{\Sigma^+}(P)\textrm{ with }\head(\psi)=\conn\}$. 

For completeness, if $\Gamma\not\der_{\bbbot}\varphi$ then $\Gamma\not\der_{\bbbotop}\varphi$ and $\Gamma\not\der_{\bbbotop}\psi$ for any $\psi$ with $\head(\psi)=\conn$. As both $\Mt$ and $\Mt^{\bbbotop}_\conn$ are saturated, we know that $\Mt\star\Mt^{\bbbotop}_\conn$ is saturated and thus there is $v\in\Val_P(\Mt\star\Mt^{\bbbotop}_\conn)$ such that $v(\Gamma)\subseteq D$, $v(\varphi)\in U$ and $v(\psi)\in U$ for every $\psi$ with $\head(\psi)=\conn$. In view of this last fact, we see that $v\in\Val_P(\Mt\star\Mt^{\bbbot}_\conn)$.

When $k=0$ there is exactly one formula whose head is $\conn$ so,  if $\Gamma\not\der_{\bbbot}\varphi$ then $\Gamma\not\der_{\bbbotop}\varphi$ and $\Gamma\not\der_{\bbbotop}\conn$, or equivalently, $\skel_\Sigma(\Gamma)\not\der_{\Mt}\skel_\Sigma(\varphi)$ and $\skel_\Sigma(\Gamma)\not\der_{\Mt}x_\conn$. Since $\Mt$ is assumed to be $2$-saturated, there is $v\in \Val_{P\cup X_\Sigma}(\Mt)$ such that $v(\skel_\Sigma(\Gamma))\subseteq D$, $v(\skel_\Sigma(\varphi))\in U$ and $v(x_\conn)\in U$. Thus, the valuation $v\circ\skel_\Sigma\in\Val_P(\Mt\star\Mt^{\bbbotop}_\conn)$ is such that $(v\circ\skel_\Sigma)(\Gamma)\subseteq D$, $(v\circ\skel_\Sigma)(\varphi)\in U$ and $(v\circ\skel_\Sigma)(\conn)\in U$. 
We conclude that $v\circ\skel_\Sigma\in\Val_P(\Mt\star\Mt^{\bbbot}_\conn)$.
\qedhere
\end{itemize}%
\end{proof}

 {
 \subsection{Translations and fibring}}
 
We close these prolegomena with some technical results concerning the relationship between the disjoint fibring of logics induced by given logical matrices, and the disjoint fibring of the logics obtained by some translations/abbreviations over those matrices. The intricacies of these results are essential for understanding how careful one needs to be when transferring examples or counterexamples to or from a combination of logics involving connectives that are defined by abbreviation. 
 From this point on, we assume fixed signatures $\Xi_1,\Xi_2,\Sigma_1,\Sigma_2$ with~$\Xi_1$ disjoint from~$\Xi_2$ and $\Sigma_1$ disjoint from~$\Sigma_2$, and translations $\bt_1:\Xi_1\longrightarrow \Sigma_1$ and $\bt_2:\Xi_2\longrightarrow \Sigma_2$. We shall write~$\Xi$ for $\Xi_1\cup\Xi_2$, $\Sigma$ for $\Sigma_1\cup\Sigma_2$, and~$\bt$ for $\bt_1\cup\bt_2$. We also fix saturated matrices~$\Mt_1$ and~$\Mt_2$ over the signatures $\Sigma_1$ and~$\Sigma_2$. In case we are given non-saturated matrices $\Mt_1$ or $\Mt_2$, we can always consider instead $\Mt_1^\omega$ or $\Mt_2^\omega$. 
 Let $\tuple{\Sigma,\der}$ represent $\fib{\cL_{{\Mt}_1}}{\cL_{{\Mt}_2}}$ and $\tuple{\Xi,\der^{\bt}}$ represent 
$\fib{\cL_{{\Mt}_1^{\bt_1}}}{\cL_{{\Mt}^{\bt_2}_2}}$. Recall that $\fib{\cL_{{\Mt}_1}}{\cL_{{\Mt}_2}}$ is characterized by $\Mt_1\star\Mt_2$, and that $\fib{\cL_{{\Mt}_1^{\bt_1}}}{\cL_{{\Mt}^{\bt_2}_2}}$ is characterized by $\Mt_1^{\bt_1}\star\Mt_2^{\bt_2}$, as $\Mt_1^{\bt_1}$ and $\Mt_2^{\bt_2}$ are both saturated (see Subsection~\ref{sec:useful}).
\medskip

\begin{proposition}\label{halflift} 
For every $\Gamma\cup\{\varphi\}\subseteq L_{\Xi}(P)$, if 
$\Gamma \der^\bt \varphi$ then $\bt(\Gamma) \der \bt(\varphi)$.%
\end{proposition}
\begin{proof}
The result follows from the fact that $\{v\circ\bt:v\in\Val_P(\Mt_1\star\Mt_2)\}\subseteq\Val_P(\Mt_1^{\bt_1}\star\Mt_2^{\bt_2})$. To see this, note that if $v\in\Val_P(\Mt_1\star\Mt_2)$ then $v\circ\bt=(\pi_1(v)\star\pi_2(v))\circ\bt=(\pi_1(v)\circ\bt_1^+)\star(\pi_2(v)\circ\bt_2^+)$ where, for $i\in\{1,2\}$, we are considering $\bt_i^+:L_{\Xi_i}(P\cup X_{\Xi_i})\longrightarrow L_{\Sigma_i}(P\cup X_{\Sigma_i})$ as an extension of $\bt_i:L_{\Xi_i}(P)\longrightarrow L_{\Sigma_i}(P)$ defined as follows: $\bt_i^+(p):=p$ for $p\in P$, $\bt_i^+(\xi(\psi_1,\ldots,\psi_k)):=\bt_i(\xi)(\bt_i^+(\psi_1),\dots,\linebreak\bt_i^+(\psi_k))$ for $\xi\in\Xi^{(k)}$, and $\bt_i^+(x_\varphi):=\skel_{\Sigma_i}(\bt(\varphi))$ for $x_\varphi\in X_{\Xi_i}$. Because $\pi_1(v)$ and $\pi_2(v)$ are compatible, it is routine to check that $(\pi_1(v)\circ\bt_1^+)\in\Val_{P\cup X_{\Xi_1}}(\Mt_1^{\bt_1})$ and $(\pi_2(v)\circ\bt_2^+)\in\Val_{P\cup X_{\Xi_2}}(\Mt_2^{\bt_2})$ are also compatible, and therefore $v\circ\bt\in\Val_P(\Mt_1^{\bt_1}\star\Mt_2^{\bt_2})$.
\end{proof}


\noindent
Note 
that the converse of the above statement is in general not true, and we can have $\bt(\Gamma) \der \bt(\varphi)$ while $\Gamma \not\der^\bt \varphi$. When this happens it must be because $\{v\circ\bt:v\in\Val_P(\Mt_1\star\Mt_2)\}\subsetneq\Val_P(\Mt_1^{\bt_1}\star\Mt_2^{\bt_2})$. 
{ 
Let $\conn$ be a binary Boolean connective, $\bot_1$ and $\bot_2$ be two $0$-place bottom-like connectives and consider $\cB_\conn\bullet\cB_{\bot_1\bot_2}$. 
Now let $\bt_1:=\mathsf{id}_\conn$, $\bt_2(\bot_1)=\bt_2(\bot_2):=\bot$, and $\bt:=\bt_1\cup \bt_2$.
Clearly $\TWO_{\bot_1\bot_2}=\TWO_\bot^{\bt_2}$ is saturated. 
Every valuation~$v$ over $\TWO^\omega_\conn\star\TWO_\bot$ is such that 
$v(\bt(\bot_1))=v(\bt(\bot_2))$, but it is not the case that $v'(\bot_1)=v'(\bot_2)$ 
for valuations $v'$ over $\TWO^\omega_\conn\star\TWO_{\bot_1\bot_2}$.
Hence, $\{v\circ\bt:v\in\Val_P(\TWO^\omega_\conn\star\TWO_\bot)\}\subsetneq\Val_P((\TWO^\omega_\conn)^{\bt_1}\star\TWO_\bot^{\bt_2})$.}  

At any rate, one may still secure the converse of the previous proposition under certain particular circumstances:

\begin{proposition}\label{liftspecial}
The following assumptions give sufficient conditions for concluding that 
$\Gamma \der^\bt \varphi$ if and only if $\bt(\Gamma) \der \bt(\varphi)$, for every $\Gamma\cup\{\varphi\}\subseteq L_{\Xi}(P)$:
\begin{itemize}
\item[$\mathbf{(a)}$] $\bt$ is injective, \underline{or}
\item[$\mathbf{(b)}$] $\top$ is the only connective in $\Sigma_2$, $\Mt_1$ is unitary and $\Mt_2=\TWO_\top$, \underline{or}
\item[$\mathbf{(c)}$] $\bot$ is the only connective in $\Xi_2=\Sigma_2$,  $\bt_2=\mathbf{id}_{\Sigma_2}$, 
and $\Mt_2=\TWO_\bot$.
\end{itemize}
\end{proposition}
\begin{proof}
In each case, we prove that $\Val_P(\Mt_1^{\bt_1}\star\Mt_2^{\bt_2})\subseteq\{v\circ\bt:v\in\Val_P(\Mt_1\star\Mt_2)\}$. Let $v'\in\Val_P(\Mt_1^{\bt_1}\star\Mt_2^{\bt_2})$.
\begin{itemize}
\item[$\mathbf{(a)}$] For $i,j\in\{1,2\}$, $i\neq j$, consider valuations $v_i\in\Val_{P\cup X_{\Sigma_i}}(\Mt_i)$ defined, by mutual recursion, as follows: $v_i(p):=\pi_i(v')(p)$, $v_i(x_{\bt(\psi)}):=\pi_i(v')(x_\psi)$, and $v_i(x_\varphi)$ is chosen compatibly with $v_j(\skel_{\Sigma_j}(\varphi))$ for $\varphi\notin\bt(L_\Xi(P))$. Note that the injectivity of $\bt$ is essential to guarantee that $v_i(x_{\bt(\psi)})=\pi_i(v')(x_\psi)$ is well defined.  
\item[$\mathbf{(b)}$] Note that $X_{\Sigma_1}=\{x_\top\}$. Consider a valuation $v_1\in\Val_{P\cup X_{\Sigma_1}}(\Mt_1)$ defined by setting $v_1(p):=\pi_1(v')(p)$, and $v_1(x_\top)$ being assigned a designated value in the only possible way, and a valuation $v_2\in\Val_{P\cup X_{\Sigma_2}}(\TWO_\top)$ defined by setting $v_2(p):=\pi_2(v')(p)$, and let the value of $v_2(x_\varphi)$ be chosen compatibly with $v_1(\skel_{\Sigma_1}(\varphi))$. Note that the unitariness of $\Mt_1$ is fundamental to the construction of $v_1$, whereas the fact that $\Mt_2=\TWO_\top$ makes compatible choices unique when constructing~$v_2$.
\item[$\mathbf{(c)}$] Note that $X_{\Sigma_1}=\{x_\bot\}$. Consider the valuation $v_1\in\Val_{P\cup X_{\Sigma_1}}(\Mt_1)$ defined by setting $v_1(p):=\pi_1(v')(p)$ and $v_1(x_\bot):=\pi_1(v')(x_\bot)$, and the valuation $v_2\in\Val_{P\cup X_{\Sigma_2}}(\TWO_\bot)$ defined by setting $v_2(p):=\pi_2(v')(p)$ and by letting the value of $v_2(x_\varphi)$ be chosen compatibly with 
$v_1(\skel_{\Sigma_1}(\varphi))$. Note again that $\Mt_2=\TWO_\bot$ makes compatible choices unique when constructing~$v_2$.
\end{itemize}
In each case, it is routine to check that~$v_1$ and~$v_2$ defined in this manner are compatible and that $v'=(v_1\star v_2)\circ\bt$.
This implies that if $\bt(\Gamma) \der \bt(\varphi)$ then $\Gamma \der^\bt \varphi$. The result then follows from Prop.~\ref{halflift}.
\end{proof}

Under the applicability conditions of the previous proposition, or in general whenever $\Gamma \der^\bt \varphi$ if and only if $\bt(\Gamma) \der \bt(\varphi)$, we have the following interesting consequences:

\begin{proposition}\label{lift}
Assume that $\Gamma \der^\bt \varphi$ if and only if $\bt(\Gamma) \der \bt(\varphi)$, for every $\Gamma\cup\{\varphi\}\subseteq L_{\Xi}(P)$. Then, the following properties hold:
\begin{itemize}
\item[--] if $\fib{\cL_{{\Mt}_1}}{\cL_{{\Mt}_2}}=\cL_\Mt$ for some matrix~$\Mt$ over~$\Sigma$ then  
$\fib{\cL_{{\Mt}_1^{\bt_1}}}{\cL_{{\Mt}^{\bt_2}_2}}=\cL_{\Mt^\bt}$, \underline{and}
\item[--] if $\fib{\cL_{{\Mt}_1}}{\cL_{{\Mt}_2}}$ is $k$-determined then so is $\fib{\cL_{{\Mt}_1^{\bt_1}}}{\cL_{{\Mt}^{\bt_2}_2}}$ for $k\in\nats$.
\end{itemize}
\end{proposition}
\begin{proof}
{
For the first property, note that  $\Val(\Mt^\bt)=\{v\circ\bt:v\in\Val(\Mt)\}$ by definition,
and therefore $\Val_P(\Mt_1^{\bt_1}\star\Mt_2^{\bt_2})\subseteq\{v\circ\bt:v\in\Val_P(\Mt_1\star\Mt_2)\}$ which, as in Prop.~\ref{liftspecial}, implies that 
$\fib{\cL_{{\Mt}_1^{\bt_1}}}{\cL_{{\Mt}^{\bt_2}_2}}=\cL_{\Mt^{\bt}}$.

For the second, we show that $k$-determinedness is preserved by $\bt$. 
 Indeed, from $\Gamma\not\der^\bt \varphi$ we obtain  $\bt(\Gamma)\not\der\bt(\varphi)$.
Assuming that $\fib{\cL_{{\Mt}_1}}{\cL_{{\Mt}_2}}$ is $k$-determined, we obtain that there is $\sigma:P\longrightarrow \{p_1,\ldots,p_k\}$ such that
$\bt(\Gamma)^\sigma\not\der\bt(\varphi)^\sigma$.
As $\sigma$ only swaps variables, and $\bt$ is the identity over variables, 
we conclude that they commute, i.e., $\bt(\psi)^\sigma=\bt(\psi^\sigma)$ for every $\psi\in L_{\Xi}(P)$.
Therefore,
$\bt(\Gamma^\sigma)\not\der\bt(\varphi^\sigma)$, and so
$\Gamma^\sigma\not\der^\bt \varphi^\sigma$. Thus $\fib{\cL_{{\Mt}_1^{\bt_1}}}{\cL_{{\Mt}^{\bt_2}_2}}$ is also $k$-determined.
}
\end{proof}
\section{Fibring disjoint fragments of classical logic}
\label{sec:merging}

\noindent 
In this section we shall establish the general results about combining Boolean connectives and, in general, fragments of classical logic. 

\subsection{Adding top-like connectives}
We start with the simplest cases where merging two fragments yields the corresponding joint fragment of classical logic, namely, when all the connectives from one of the given fragments are top-like.

\begin{proposition} 
\label{topisnice}
If the signatures $\Sigma_1$ and $\Sigma_2$ are disjoint and $\cC_2^{\Sigma_2}\subseteq\cC_2^\top$ then  
$\fib{\cB_{\Sigma_1}\!}{\cB_{\Sigma_2}}=\cB_{\Sigma_1\cup\Sigma_2}$.
\end{proposition}
\begin{proof}
{Consider a $0$-place top-like connective $\top$ such that $\top\notin\Sigma_1$ (if $\top \in \Sigma_1$ we pick a syntactically different copy).}
We obtain that 
  $\fib{\cB_{\Sigma_1}\!}{\cB_{\top}}=\cB_{\Sigma_1\cup\{\top\}}$, as a simple corollary of Prop.~\ref{withtop}. Just note that because $\cB_{\Sigma_1}$ is characterized by $\TWO_{\Sigma_1}$, and $\cB_{\top}=\bbtop_{\top}$ and $\Mt^{\bbtop}_{\top}=\TWO_{\top}$, we have that $\fib{\cB_{\Sigma_1}\!}{\cB_{\top}}$ is characterized by $\TWO_{\Sigma_1}\star\TWO_\top$.
Further, it is immediate to see that $\TWO_{\Sigma_1}\star\TWO_\top$ is isomorphic to $\TWO_{\Sigma_1\cup\{\top\}}$.
{Since $\cB_{\Sigma_1}$ is characterized by the unitary matrix $\TWO_{\Sigma_1}$ and $\cB_{\top}=\cL_{\TWO_\top}$, 
we are under the applicability conditions of Prop.~\ref{liftspecial}$\mathbf{(b)}$.
{Let $\bt$ be the translation that sends every connective of $\Sigma_2$ to $\top$.}
Hence, as  $\fib{\cB_{\Sigma_1}\!}{\cB_{\top}}$ 
is characterized by the matrix $\TWO_{\Sigma_1\cup\{\top\}}$, we conclude by Prop.~\ref{lift}
that
$\cB_{\Sigma_1}\bullet\cB_{\Sigma_2}$ is characterized by $\TWO_{\Sigma_1\cup\{\top\}}^{\mathbf{id}_{\Sigma_{\!1}}\!\cup\bt}=\TWO_{\Sigma_1\cup\Sigma_2}$, and therefore $\fib{\cB_{\Sigma_1}\!}{\cB_{\Sigma_2}}=\cB_{\Sigma_1\cup\Sigma_2}$.}
\end{proof}


\begin{example}[Coimplication and top]\label{coimptop} 
\em
Consider adding classical coimplication $\dcoimp{}{}$ to top $\top$, that is, fibring the logics~$\cB_{\dcoimpscript{}{}}$ and~$\cB_\top$. 
Recall from Subsection~\ref{classical}  the semantics and axiomatization of~$\cB_\top$. 
Coimplication is characterized by the $2$-valued matrix $2_{\dcoimpscript{}{}}$ where:
 \begin{center}
  \begin{tabular}{c | c c}
    $\widetilde{\smash{\dcoimp{}{}}}$ & $0$  & $1$    \\ 
    \hline
    $0$ & $0$ & $1$ \\ 
     $1$  &$0$ & $0$  
  \end{tabular}\end{center}
\noindent
We shall not explicitly provide here a Hilbert calculus for~$\cB_{\dcoimpscript{}{}}$. The methods in~\cite{Rautenberg1981} would allow one to obtain such a calculus, but the general procedure is tedious and we leave it to the interested reader. 
We note that $\TWO_{\dcoimpscript{}{}}$ is not saturated: note for instance that 
$p\not\der_{\dcoimpscript{}{}} \dcoimp{p}{q}$ and $p\not\der_{\dcoimpscript{}{}} q$, but no bivaluation can set, at the same time, $v(p)=1$ and $v(\dcoimp{p}{q})=v(q)=0$. 
However, in this case, we can rely on Prop.~\ref{withtop}, or more generally on Prop.~\ref{topisnice}, to conclude that 
$\cB_{\dcoimpscript{}{}}\bullet\cB_\top=\cB_{\dcoimpscript{}{} \top}$ is characterized by the matrix 
$\TWO_{\dcoimpscript{}{}}\star \TWO_\top=\TWO_{\dcoimpscript{}{} \top}$. This is, of course, a very special case, also because $\{\dcoimp{}{},\top\}$ forms a functionally complete set of classical connectives (in fact, it is functionally complete \textit{in a stronger sense}, as it also allows for the standard definition of the $0$-place Boolean operations --- see Section~3.14 of~\cite{Humb:Conn}).\hfill$\triangle$
\end{example}

\subsection{When none of the connectives is very significant}

Another case where fibring yields the corresponding classical fragment comes about when all the connectives involved fail to be very significant.

\begin{proposition} 
\label{notverydifficult}
If the signatures $\Sigma_1$ and $\Sigma_2$ are disjoint and $\cC_2^{\Sigma_1},\cC_2^{\Sigma_2}\subseteq\cC_{2}^{\land\top\bot}$ then  
$\fib{\cB_{\Sigma_1}\!}{\cB_{\Sigma_2}}=\cB_{\Sigma_1\cup\Sigma_2}$.
\end{proposition}
\begin{proof}
We know from Prop.~\ref{satveysig} that $\TWO_{\Sigma_1\!}$ and $\TWO_{\Sigma_2\!}$ are saturated, since the connectives are not very significant. Hence, it follows from the results mentioned in Rem.~\ref{sem-fib} that $\fib{\cB_{\Sigma_1\!}}{\cB_{\Sigma_2}}$ is characterized by ${\TWO_{\Sigma_1\!}\star \TWO_{\Sigma_2\!}}$. To conclude, just observe that $\TWO_{\Sigma_1}\star\TWO_{\Sigma_2}$ is isomorphic to $\TWO_{\Sigma_1\cup\Sigma_2}$.
\end{proof}

{%
In particular, this implies that if we merge the axiomatizations of two projec\-tion-conjunctions with the same arity we obtain a logic in which these connectives {collapse}.

\begin{example}[Two copies of conjunction]\label{ex:conj}
\em
We will now consider two syntactically distinct copies, say~$\e$ and~$\&$, of conjunction, that is, we will combine through fibring two copies of the conjunction-only fragment of classical logic, $\cB_\e$ and $\cB_\&$. Semantically, they are characterized by the matrices~$\TWO_\e$ and~$\TWO_\&$ with $\widetilde{\&}:=\widetilde{\e}$ defined as in Subsection~\ref{classical}. A Hilbert calculus for $\cB_\&$ is a simple translated copy of the one provided for~$\cB_\&$ in {[$\cB_\land$]}, \textit{mutatis mutandis}.

By Prop.~\ref{satveysig}, both 
 $2$-valued matrices are saturated since conjunctions are not very significant. Indeed, $\Gamma\der_{\cB_\e}\varphi$ precisely when $\var(\varphi)\subseteq\var(\Gamma)$, a non-trivial theory $\Gamma^{\der_{\cB_\e}}$ is characterized by the bivaluation $v$ such that $v(p)=1$ if $p\in\var(\Gamma)$, and $v(p)=0$ if $p\notin\var(\Gamma)$.

In view of Rem.~\ref{rem-significant} and the results mentioned in Rem.~\ref{sem-fib}, or more generally in view of Prop.~\ref{notverydifficult}, it is clear that $\cB_\e\bullet\cB_\&$ is characterized by $\TWO_\e\star \TWO_\&$, which is again $2$-valued and where (up to isomorphism) $\widetilde{\conn}_\star=\widetilde{\conn}$ for $\conn\in\{\e,\&\}$. Clearly, this means that the two conjunctions collapse and that $\cB_\e\bullet\cB_\&=\cB_{\e\&}$.
Consequently, a complete calculus for $\cB_{\e\&}$ is obtained by just merging the calculi for the components.
%
\hfill$\triangle$
\end{example} 

}


\subsection{Non-finitely-valued combinations}

We now start to establish the negative cases, that is, to identify the situations when the fibring of classical connectives results in a logic that is subclassical. 

\begin{proposition} 
\label{prop:failingmaya}
The fibring $\fib{\cB_{\conn_1\!}}{\cB_{\conn_2}}$ of the logic of a very significant Boolean connective~$\conn_1$ and 
the logic of a 
non-top-like Boolean connective~$\conn_2$ distinct from $\bot$
fails to be locally tabular, and therefore $\fib{\cB_{\conn_1\!}}{\cB_{\conn_2}}\subsetneq \cB_{\conn_1 \conn_2}$. 
\end{proposition}
\begin{proof}
In order to show that $\fib{\cB_{\conn_1\!}}{\cB_{\conn_2}}$ is not locally tabular, we shall build an infinite collection $\{\varphi_t\}_{t\in\nats}$ of formulas in $L_{\conn_1\conn_2}(P)$, using only finitely many distinct sentential variables, and then show them to be pairwise non-equivalent.

Let us first focus on $\conn_1$. Recall that $\cB_{\conn_1\!}$ is characterized by the saturated matrix $\TWO^\omega_{\conn_1}$.
Let~$\conn_1$ be a $k$-place very significant connective with $j< k$ projective indices. We assume without loss of generality that the projective indices of~$\conn_1$ correspond to its first~$j$ arguments. Let $s=k-j$. As in the proof of Prop.~\ref{satveysig}, we have: 
\begin{itemize}\setlength\itemsep{0pt}
\item[(a)] $\conn_1(p_1,\dots,p_j,x_1,\dots,x_s)\der_1 p_i$ for $1\leq i\leq j$;
\item[(b)] $\conn_1(p_1,\dots,p_j,x_1,\dots,x_s)\not\der_1 x_i$ for $1\leq i\leq s$;
\item[(c)] $\conn_1(p_1,\dots,p_j,x_1,\dots,x_s)\not\der_1\conn_1(p_1,\dots,p_j,y_1,\dots,y_s)$;
\item[(d)] $\conn_1(p_1,\dots,p_j,x_1,\dots,x_s)\not\der_1 y_i$ for $1\leq i\leq s$.
\end{itemize}
From (a)--(d), taking into consideration the theory $\{\conn_1(p_1,\dots,p_j,x_1,\dots,x_s)\}^{\der_1}$, we may conclude that there is a valuation $v_1$ over $\TWO^\omega_{\conn_1}$ such that 
  $v_1(\conn_1(p_1,\dots,p_j, 
  \linebreak
  x_1,\dots,x_s))=v_1(p_i)=\nats$ for 
$1\leq i\leq j$, $v_1(x_i)\neq\nats$ and $v_1(y_i)\neq\nats$ for $1\leq i\leq s$, and $v_1(\conn_1(p_1,\dots,p_j,y_1,\dots,y_s))\neq\nats$.

Next, on what concerns $\conn_2$, recall from Rem.~\ref{rem-significant} that a non-top-like Boolean connective distinct from~$\bot$ cannot be $0$-place. Hence, according to Lemma~\ref{lem:nontop}, we can fix a non-top-like compound 1-place $\theta\in L_{\conn_2}(\{p\})$. Further, we know from the latter lemma that:
\begin{itemize}
\item[(e)] $\not\der_2\theta^n(p)$ for every $n\in\nats$.
\end{itemize}
As $\cB_{\conn_2\!}$ is characterized by the saturated matrix $\TWO^\omega_{\conn_2}$, from (e), considering the theory $\varnothing^{\der_2}$ we conclude that there exists a valuation $v_2$ over $\TWO^\omega_{\conn_2}$ such that $v_2(\theta^n(p))\neq\nats$ for every $n\in\nats$.

Let us finally consider the following formulas on $j+1$ sentential variables: 
$$\varphi_t:=\conn_1(p_1,\dots,p_j,\theta^{1+ts}(p),\dots,\theta^{(t+1)s}(p)),\text{ for $t\in\nats$}$$ 
In these formulas, we sequentially deploy $s$ distinct nestings of $\theta$ on the sentential variable $p$, in the positions corresponding to non-projective components of $\conn_1$.

Take $t_1\neq t_2$. We will show that $\varphi_{t_1}{\EqDiv{}{}{}}\varphi_{t_2}$ fails to hold, taking advantage of the completeness of the saturated \nmatrix{} 
${\TWO^\omega_{\conn_1}}{\star}{\TWO^\omega_{\conn_2}}$ for $\fib{\cB_{\conn_1\!}}{\cB_{\conn_2}}$. For that purpose, consider 
$\Gamma:=\{\varphi_{t_1},\varphi_{t_2}\}\cup\{\theta^{i+t_1s},\theta^{i+t_2s}:1\leq i\leq s\}$, and let $x_1,\dots,x_s,y_1,\dots,y_s$ be the sentential variables in $X_{\Sigma_1}$ such that $x_i:=x_{\theta^{i+t_1s}(p)}=\skel_{\Sigma_1}(\theta^{i+t_1s}(p))$ and  $y_i:=x_{\theta^{i+t_2s}(p)}=\skel_{\Sigma_1}(\theta^{i+t_2s}(p))$ for $1\leq i\leq s$. 

As the mapping $v_1$ is not defined for $p$ nor for the special sentential variables~$x_{\psi}$, for $\psi\in\sub(\{\theta^{i+t_1s},\theta^{i+t_2s}:1\leq i\leq s\})
\setminus\{p\}$, but these variables also do not occur in $\skel_{\Sigma_1}(\Gamma)$, we can extend~$v_1$ to a 
$\skel_{\Sigma_1}(\sub(\Gamma))$-partial valuation $v'_1$ such that $v'_1(p)$ and each $v'_1(x_\psi)$ are assigned designated values, respectively, if and only if~$v_2(p)$ and~$v_2(\psi)$ are assigned designated values.

Similarly, $v_2$ is not defined for $p_1,\dots,p_j$ nor for $x_{\varphi_{t_1}}, x_{\varphi_{t_2}}$, and these variables do not occur in $\skel_{\Sigma_2}(\Gamma)$, so we can extend~$v_2$ to a $\skel_{\Sigma_2}(\sub(\Gamma))$-partial valuation $v'_2$ such that $v'_2(p_i)$, for each $1\leq i\leq j$, and $v'_2(x_{\varphi_{t_1}}),v'_2(x_{\varphi_{t_2}})$ are chosen to be compatible, respectively, with $v_1(p_i)$ and $v_1(\varphi_{t_1}),v_1(\varphi_{t_2})$.

It is clear that $v'_1$ and $v'_2$ satisfy the compatibility requirement of Lemma~\ref{mergingvals},
and therefore the $\Gamma$-partial valuation $v'_1\star v'_2$ over ${\TWO^\omega_{\conn_1}}{\star}{\TWO^\omega_{\conn_2}}$ does the job.
As it is clear that $\fib{\cB_{\conn_1\!}}{\cB_{\conn_2}}\subseteq \cB_{\conn_1 \conn_2}$, and also that $\cB_{\conn_1 \conn_2}$ is locally tabular, we conclude that $\fib{\cB_{\conn_1\!}}{\cB_{\conn_2}}\subsetneq \cB_{\conn_1 \conn_2}$.
\end{proof}


%
%

{


We conclude from the above, in contrast to what happens with conjunction (Ex.~\ref{ex:conj}), that when we merge the axiomatizations of two copies of a very significant connective we obtain a logic where these two copies do not collapse.

\begin{example}[Two copies of disjunction]\label{ex:disj}
\em 
This time let us consider two syntactically distinct copies, say $\ou$ and $||$, of disjunction, that is, let us fiber two copies of the disjunction-only fragment of classical logic, $\cB_\ou$ and~$\cB_{||}$. 
For an illustration of the construction in the proof of Prop.~\ref{prop:failingmaya} in the case of $\cB_\ou\bullet\cB_{||}$, note that the formulas $\theta^1(p)\ou\theta^2(p)$, $\theta^3(p)\ou\theta^4(p)$, $\theta^5(p)\ou\theta^6(p)$, $\dots$, where $\theta(p)=p || p$, are all pairwise non-equivalent.

Semantically, the above mentioned logics are characterized by the matrices $\TWO_\ou$ and $\TWO_{||}$ with $\widetilde{||}:=\widetilde{\ou}$ defined as in Subsection~\ref{classical}. A Hilbert calculus for $\cB_{||}$ is a simple translated copy of the one provided in {[$\cB_\ou$]}.
Again, it is easy to see that the $2$-valued matrices are not saturated. For instance, $p\,\ou\,q\not\der_{\cB_\ou} p$ and $p\,\ou\,q\not\der_{\cB_\ou} q$, but no bivaluation can set $v(p\,\ou\,q)=1$ and at the same time $v(p)=v(q)=0$. 

It follows from the results mentioned in Rem.~\ref{sem-fib} that $\cB_\ou\bullet\cB_{||}$ is characterized by the strict product of the saturations $\TWO^\omega_{\ou},\TWO^\omega_{||}$, the non-denumerably large \nmatrix{} defined (up to isomorphism) by  
$\TWO^\omega_{\ou}\star \TWO^\omega_{||}=\tuple{V,\{(\nats,\nats)\},
\cdot_\star}$ 
where:

\begin{itemize}
\item[] $V:=\{(X,Y): X,Y\subseteq \nats\mbox{ and }X=\nats \mbox{ iff } Y=\nats \}$
\item[] $(X_1,Y_1)\widetilde{\ou}_\star (X_2,Y_2):=\begin{cases}
\{(\nats,\nats)\}, & \mbox{ if }X_1\cup X_2=\nats\\
\{(X_1\cup X_2, Y):Y\subsetneq \nats\},& \mbox{ if }X_1\cup X_2\neq\nats
		\end{cases}$
\item[] $(X_1,Y_1)\widetilde{||}_\star (X_2,Y_2):=\begin{cases}
\{(\nats,\nats)\}, & \mbox{ if }Y_1\cup Y_2=\nats\\
\{(X,Y_1\cup Y_2):X\subsetneq \nats\},& \mbox{ if }Y_1\cup Y_2\neq\nats
		\end{cases}$
\end{itemize}

\noindent
This is not unexpected, as classical disjunction is a very significant connective, and therefore $\cB_\ou\bullet\cB_{||}$ is known to be non-finitely-valued, as a consequence of Prop.~\ref{prop:failingmaya}. Thus, $\cB_\ou\bullet\cB_{||}$ is strictly weaker than $\cB_{\ou\,||}$, the two disjunctions do not collapse --- for instance, the mixed consequence assertion $p\ou q\der p\,||\,q$ fails to hold in $\cB_\ou\bullet\cB_{||}$. The latter logic cannot even be said to be finitely-\nvalued{}, as we can indeed show that it fails to be $k$-determined for any $k\in\nats$ (recall Subsection~\ref{tabularity}). To see this, consider:
\begin{align*}
	\Gamma_k :=\{p_i\ou p_j:1\leq i<j\leq k+1\} \mbox{, and }
	\varphi_k:=\bigvee\limits_{1\leq i \leq k+1} q\,||\,(p_i\ou q).
\end{align*}
  It is clear that for every $\sigma:P\longrightarrow  \{p_1,\ldots,p_k\}$ we have that $\sigma(p_i)=\sigma(p_j)$ for some $1\leq i< j\leq k+1$.
Hence, it is straightforward to conclude in this case that
(i) $\Gamma_k^\sigma\der \sigma(p_i)\ou\sigma(p_j)$, (ii) $\sigma(p_i)\ou\sigma(p_j)\der\sigma(p_i)$, (iii) $\sigma(p_i)\der (p_i\ou q)^\sigma$, and (iv) $(p_i\ou q)^\sigma\der \varphi_k^\sigma$, and from these it immediately follows that $\Gamma_k^\sigma\der\varphi_k^\sigma$.
Now, to show that $\Gamma_k\not\der\varphi_k$, just consider a valuation $v$ on $\TWO^\omega_{\ou}\star \TWO^\omega_{||}$ such that: 
\begin{align*}
v(q)&=(\nats\setminus \{1,\ldots,k+1\},\varnothing),\\
 v(p_i)=v(p_i\ou q)&=(\nats\setminus \{i\},\varnothing), \mbox{ for } 1\leq i\leq k+1,\\
 v(p_i\vee p_j)&=(\nats,\nats), \mbox{ for } 1\leq i<j\leq k+1,\\
v(q\,||\,(p_i\ou q))&=(\varnothing,\varnothing), \mbox{ for } 1\leq i\leq k+1,\\
v(\bigvee\limits_{0\leq i \leq \ell} q\,||\,(p_i\ou q))&=(\varnothing,\varnothing),  \mbox{ for  } \ell\leq k+1.
\end{align*}
  
%

The merged axiomatization for $\cB_\ou\bullet\cB_{||}$ is built as usual. 
More interestingly, after~\cite{Rautenberg1981}, note that a complete Hilbert calculus for $\cB_{\ou ||}$ may be obtained more simply by adding the following interaction rules to the Hilbert calculus given to~$\ou$ in [$\cB_\ou$]:

\begin{center}
\scalebox{1.2}{
\begin{tabular}{cccccccc}
&$\frac{p\,\ou\,(q\,\ou\,r)}{p\,\ou\,(q\,||\,r)}$&$\frac{p\,\ou\,(q\,||\,r)}{p\,\ou\,(q\,\ou\,r)}$
\end{tabular}
}
\end{center}
All the translated rules for the disjunction~$||$ are easily derivable from the latter mentioned rules.\smallskip

Note that what we said about merging two copies of the Boolean disjunction applies \emph{mutatis mutandis} to the case of two copies of the Boolean implication. The reason is that classical implication is known to express classical disjunction, e.g., via a translation {$\bt(\ou)=\lambda p_1p_2.\,\imp{(\imp{p_1}{p_2})}{p_2}$}. 
\hfill$\triangle$
\end{example} 

An equally interesting non-collapsing example is provided by merging the axiomatizations of two copies of classical negation:

\begin{example}[Two copies of negation]\label{twonegs}
\em
We will now combine~$\cB_\neg$ and~$\cB_\sim$ through fibring, where $\neg$ and~$\sim$ are two syntactically distinct copies of classical negation. 
Semantically, they are characterized by the matrices~$\TWO_\neg$ and~$\TWO_\sim$ with $\widetilde{\sim}:=\widetilde{\neg}$ as defined in Subsection~\ref{classical}. A Hilbert calculus for~$\cB_\sim$ is a 
simple translated copy of the one provided in {[$\cB_\neg$]}.

It is easy to see now that the $2$-valued classical matrices are not saturated. For instance, $\not\der_{\cB_\neg} p$ and $\not\der_{\cB_\neg} \neg p$, but no bivaluation can fail to satisfy both non-theorems simultaneously, that is, setting $v(p)=v(\neg p)=0$ is impossible. 

In any case, it follows from the results mentioned in Rem.~\ref{sem-fib} that $\cB_\neg\bullet\cB_\sim$ is characterized by $\TWO^\omega_\neg\star \TWO^\omega_\sim$, a non-denumerably large \nmatrix{}. This is not too bad, as classical negation is a very significant connective, and therefore $\cB_\neg\bullet\cB_\sim$ is not finitely-valued, 
as a consequence of Prop.~\ref{prop:failingmaya}.
Thus, $\cB_\neg\bullet\cB_\sim$ is strictly weaker than $\cB_{\neg\sim}$, and the two negations do not collapse --- for instance, the mixed consequence assertion $\neg p\der{\sim} p$ fails to hold in $\cB_\neg\bullet\cB_\sim$.

A further interesting fact about this particular example is that $\cB_\conn$, for $\conn\in\{\neg,\sim\}$, turns out to have an alternative semantic characterization by way of the $3$-valued deterministic matrix 
$\Mt^3_\conn:=\tuple{\{0,\frac{1}{2},1\},\{1\},\widetilde{\conn}_3}$
where:
 \begin{center}
  \begin{tabular}{c | c c}
    &  $\widetilde{\conn}_3$   \\ 
    \hline
    $0$ & $1$   \\ 
        $\frac{1}{2}$ & $\frac{1}{2}$  \\ 
     $1$  &$0$    
  \end{tabular}\end{center}
\noindent 
What is more, this $3$-valued matrix is saturated. Indeed, since $\Gamma\der_{\cB_\neg}\neg^i\varphi$ iff $\Gamma\der_{\cB_\neg}\varphi$ for $i$ even, or if $\Gamma\der_{\cB_\neg}\neg\varphi$ for~$i$ odd, a non-trivial theory $\Gamma^{\der_{\cB_\neg}}$ is precisely characterized by the valuation $v$ such that $v(p)=1$ if $p\in\Gamma^{\der_{\cB_\neg}}$, $v(p)=0$ if $\neg p\in\Gamma^{\der_{\cB_\neg}}$, and $v(p)=\frac{1}{2}$ if $p\notin\Gamma^{\der_{\cB_\neg}}$ and $\neg p\notin\Gamma^{\der_{\cB_\neg}}$.

Now, in view of the facts mentioned in Rem.~\ref{sem-fib}, 
it follows that $\cB_\neg\bullet\cB_\sim$ is also semantically characterized by the $5$-valued \nmatrix{} defined by 
$\Mt^3_\neg\star \Mt^3_\sim=\tuple{\{(0,0),(0,\frac{1}{2}),(\frac{1}{2},0),(\frac{1}{2},\frac{1}{2}),(1,1)\},\{(1,1)\},\widetilde{\cdot}_5}$ where:
 \begin{center}
  \begin{tabular}{c | c |c}
    &  $\widetilde{\neg}_5$  &  $\widetilde{\sim}_5$ \\ 
    \hline
    $(0,0)$ & $\{(1,1)\}$ & $\{(1,1)\}$   \\[.4mm] 
    $(0,\frac{1}{2})$ & $\{(1,1)\}$ & $\{(0,\frac{1}{2}),(\frac{1}{2},\frac{1}{2})\}$  \\[.4mm]  
    $(\frac{1}{2},0)$ & $\{(\frac{1}{2},0),(\frac{1}{2},\frac{1}{2})\}$ & $\{(1,1)\}$  \\[.4mm]  
    $(\frac{1}{2},\frac{1}{2})$ & $\{(\frac{1}{2},0),(\frac{1}{2},\frac{1}{2})\}$ & $\{(0,\frac{1}{2}),(\frac{1}{2},\frac{1}{2})\}$  \\[.4mm]  
     $(1,1)$  &$\{(0,0),(0,\frac{1}{2})\}$  &$\{(0,0),(\frac{1}{2},0)\}$  
  \end{tabular}\end{center}
On the one hand, an axiomatization for $\cB_\neg\bullet\cB_\sim$ is obtained by merging the calculi for the components.
%
%
On the other hand, a complete calculus for $\cB_{\neg\sim}$ may be obtained by adding to the mentioned axiomatization for $\cB_\neg\bullet\cB_\sim$ the following interaction rules:
\begin{center}
\scalebox{1.2}{
\begin{tabular}{cccccccc}
$\frac{\;\neg p\;}{\;\; \sim p\;\;}$&&&$\frac{\;\sim p\;}{\;\; \neg p\;\;}$\\[.4mm]
\end{tabular}
}
\end{center}
Completeness of the resulting calculus may easily be confirmed with the help of Lemma~\ref{addrules}(a). 
It is indeed straightforward to see that the valuations on $\Mt^3_\neg\star \Mt^3_\sim$ respecting the two above mentioned interaction rules cannot use the values 
$(0,\frac{1}{2})$ and $(\frac{1}{2},0)$.
 Purging the $5$-valued \nmatrix{} from these values we obtain a (deterministic!)\ \nmatrix{} that is isomorphic to~$\Mt^3_\conn$ on both components. \hfill$\triangle$
\end{example} 
}

{
\noindent 
Prop.~\ref{prop:failingmaya} also happens to be informative when we combine distinct connectives:

\begin{example}[Conjunction and disjunction]\label{ex:conjdisj}
\em
We will now add classical conjunction~$\e$ to classical disjunction~$\ou$, that is, we will combine through fibring the logics~$\cB_\e$ and~$\cB_\ou$. Recall from Subsection~\ref{classical} the semantics and axiomatizations of the latter logics. We have seen in Ex.~\ref{ex:conj} and Ex.~\ref{ex:disj} that $\TWO_\e$ is saturated but $\TWO_\ou$ is not. 
From the results mentioned in Rem.~\ref{sem-fib}, it follows that $\cB_\e\bullet\cB_\ou$ is characterized by the strict product of~$\TWO_{\e}$ and~$\TWO^\omega_{\ou}$, the non-denumerable \nmatrix{} defined (up to isomorphism) by  
$\TWO_{\e}\star \TWO^\omega_{\ou}=\tuple{2^\nats,\{\nats\},\widetilde{\cdot}_\star}$ where:
\begin{itemize}
\item[] $X\widetilde{\e}_\star Y:=\begin{cases}
\{\nats\}, & \mbox{ if }X=Y=\nats\\
\{Z:Z\subsetneq \nats\},& \mbox{ if }X\cap Y\neq\nats
		\end{cases}$
\item[] $X\widetilde{\ou}_\star Y:=X\cup Y$
\end{itemize}

\noindent
Given that classical disjunction is a very significant connective, and that classical conjunction is not top-like, as a consequence of Prop.~\ref{prop:failingmaya} we have that the fibred logic is not finitely-valued. We actually also know from~\cite{charac-conserv:JLC} that this logic is not finitely-\nvalued{}. Clearly, $\cB_\e\bullet\cB_\ou$ is subclassical and, for instance, $p\ou (q\e r)\not\der (p\,\ou\,q)\e (p\,\ou\,r)$.

An axiomatization for $\cB_\e\bullet\cB_\ou$ may be obtained as usual. More interestingly, after~\cite{Rautenberg1981}, a complete calculus for $\cB_{\e\ou}$ may be obtained by simply adding three interaction rules to the calculus of disjunction, namely:

\begin{center}
\scalebox{1.2}{
\begin{tabular}{ccccccc}
$\frac{\;p\ou q\;\;\;p\ou r\;}{\;p\ou (q\e r)\;}$ &$\frac{\;p\ou (q\e r)\;}{\;p\ou q\;}$ &$\frac{\;p\ou (q\e r)\;}{\;p\ou r\;}$
\end{tabular}
}
\end{center}
%
All the rules of conjunction are derivable from the latter mentioned rules.\hfill$\triangle$
%
%
%
%
%
%
%
%
%
\end{example}

We finish illustrating Prop.\ref{prop:failingmaya} with a combination of two classical connectives that results functionally complete:

\begin{example}[Disjunction and negation]
\label{ex:disjneg}
\em
We now consider adding classical disjunction $\ou$ to classical negation $\neg$, that is, fibring the logics~$\cB_\ou$ and~$\cB_\neg$. 
Recall from Subsection~\ref{classical}  the corresponding semantics and axiomatizations of the latter. We have seen in Ex.~\ref{ex:disj} and Ex.~\ref{twonegs} that neither 
$\TWO_\ou$ nor $\TWO_\neg$ are saturated. However, we can consider the $3$-valued saturated matrix $\Mt^3_\neg$ instead of $\TWO_\neg$. Again, it follows from the results mentioned in Rem.~\ref{sem-fib} that
$\cB_\ou\bullet\cB_\neg$ is characterized by the non-denumerable \nmatrix{} 
$\TWO^\omega_{\ou}\star \Mt^3_\neg=\tuple{V,\{(\nats,1)\},\widetilde{\cdot}_\star}$. 
We leave the details of the verification to the interested reader.
%
%
%
%
%
%
%
As classical disjunction is very significant and classical negation is not top-like,  Prop.~\ref{prop:failingmaya} implies that the combined logic is not finitely-valued. We have further shown in~\cite{charac-conserv:JLC} that this logic is not finitely-\nvalued{}. Of course, 
$\cB_\ou\bullet\cB_\neg$ is subclassical and, for instance, 
$\not\der p\ou \neg p$.

The merged axiomatization for $\cB_\ou\bullet\cB_\neg$ is obtained as usual. More interestingly, again after~\cite{Rautenberg1981}, a complete calculus for $\cB_{\ou\neg}$ may be obtained by simply adding the following four interaction rules to the calculus of disjunction:
\begin{center}
\scalebox{1.2}{
\begin{tabular}{ccccccc}
$\frac{}{\;p\ou\neg p\;}$ &$\frac{\;p\ou q\;}{\;p\ou \neg\neg q\;}$&$\frac{\;p\ou \neg\neg q\;}{\;p\ou q\;}$ &$\frac{\;p\ou q\quad p\ou\neg q\;}{\;p\;}$
\end{tabular}
}
\end{center}
\noindent
The rules of negation are derivable from these. 
%
%
%

The present example has the additional interest that $\{\ou,\neg\}$ forms a functionally complete set of classical connectives, and we obtain thus from the above an axiomatization of full classical logic.\hfill$\triangle$
\end{example}
}

\subsection{Adding the connective $\bot$}

At this point, we are just left with the problem of categorizing combinations involving the $0$-place connective~$\bot$. 
We start by showing that all disjoint fibrings of a fragment of classical logic with $\bot$ are $4$-\nvalued{}:

\begin{proposition}\label{withbot}
If $\bot\notin\Sigma$ then $\cB_\Sigma\bullet \cB_\bot$ is $4$-\nvalued{}. 
\end{proposition}
\begin{proof}
This is a simple corollary of Prop.~\ref{withtop}. Note that $\cB_\Sigma$ is characterized by the matrix $\TWO_\Sigma$, and 
$\cB_\bot=\bbbot_\bot$ is characterized by the matrix $\Mt^{\bbbot}_\bot=\TWO_\bot$. As $\bot$ is a $0$-place connective we need no more than $2$-saturation. Hence, $\cB_\Sigma\bullet \cB_\bot$ is characterized by the $4$-valued \nmatrix{} $\TWO_\Sigma^2\star\TWO_\bot$.
\end{proof}

\begin{example}[Coimplication and bottom]\label{coimpbot}
\em
Recall coimplication $\dcoimp{}{}$ from Sec.~\ref{classical}.
When fibring $\cB_{\dcoimpscript{}{}}$ and $\cB_{\bot}$, we 
  make use of the general recipe in Prop.~\ref{withbot}, 
  which 
  shows that $\cB_{\dcoimpscript{}{}}\bullet \cB_\bot$ is characterized by the $4$-valued \nmatrix{}
$\TWO^2_{\dcoimpscript{}{}}\star \TWO_\bot:=\tuple{\{0,1\}^2,\{(1,1)\},\widetilde{\cdot}_\star}$
where:
\begin{center}
    \begin{minipage}{0.55\textwidth}
        \centering
            \begin{tabular}{c | c c c c}
    $\widetilde{\smash{\dcoimp{}{}}}_\star$ & $(0,0)$ & $(0,1)$ & $(1,0)$ & $(1,1)$\\ 
    \hline
    $(0,0)$ & $(0,0)$ & $(0,1)$  & $(1,0)$ & $(1,1)$ \\  
     $(0,1)$  &$(0,0)$ & $(0,0)$  & $(1,0)$ & $(1,0)$ \\  
   $(1,0)$ & $(0,0)$ & $(0,1)$ & $(0,0)$ & $(0,1)$\\
   $(1,1)$ & $(0,0)$ & $(0,0)$ & $(0,0)$ & $(0,0)$
     \end{tabular}
    \end{minipage}
        \begin{minipage}{0.25\textwidth}
        \centering
          \begin{tabular}{c c c}
    $\widetilde{\bot}_\star$   \\ 
    \hline
    $(0,0),(0,1),(1,0)$\\ \\ \\
    \mbox{}
  \end{tabular}
    \end{minipage}%

\end{center}
\noindent 
Note that the non-determinism is again concentrated on $\bot$. 
 Furthermore, the combined logic $\cB_{\dcoimpscript{}{}}\bullet \cB_\bot$ fails the cancellation property:
 $\dcoimp{\,q}{\bot}, p \der \dcoimp{\,p}{\bot}$ yet $p \not\der \dcoimp{\,p}{\bot}$. So, $\cB_{\dcoimpscript{}{}}\bullet \cB_\bot$ is not many-valued, and it is therefore subclassical.

A complete calculus for $\cB_{\dcoimpscript{}{}\bot}$ may be obtained by adding to a calculus for~$\cB_{\dcoimpscript{}{}}$ the  single interaction rule:
	\begin{center}
\scalebox{1.2}{
\begin{tabular}{ccccccc}
$\frac{p}{\;\;\dcoimpscript{\,p}{\bot}}$\\[.4mm]
\end{tabular}
}
\end{center}
\noindent
{
Completeness of the resulting calculus} may be confirmed using Lemma~\ref{addrules}(a). 
However, note that Lemma~\ref{addrules} demands the original \nmatrix{} $\Mt$ to be saturated
 in order to guarantee that the restriction that its proof promotes on the set of valuations gives
 a complete semantics for any strengthening of~$\cL_\Mt$.  
Hence, as the matrix of coimplication is not saturated we cannot consider $\TWO^2_{\dcoimpscript{}{}}\star \TWO_\bot$. 
However, from the results mentioned in Rem.~\ref{sem-fib}, we can consider $\TWO^\omega_{\dcoimpscript{}{}}\star \TWO_\bot$ knowing that the underlying matrix is saturated and also characterizes $\cB_{\dcoimpscript{}{}}\bullet \cB_\bot$.
We thus have that $\TWO^\omega_{\dcoimpscript{}{}}\star \TWO_\bot=\tuple{2^\nats,\{\nats\},\widetilde{\cdot}_\star}$ with
$X\widetilde{\smash{\dcoimp{}{}}}_\star Y=\overline{X}\cap Y$
and
$\widetilde{\bot}_\star=2^\nats\setminus\{ \nats\}$. 
%
%
To conclude the argument, it is enough to show that for every valuation over $\TWO^\omega_{\dcoimpscript{}{}}\star \TWO_\bot$ that respects the above interaction rule 
there is a valuation over a Boolean matrix that satisfies the same formulas.
For that effect there are two cases to analyze.
Clearly, every $v$ that fails to satisfy all formulas in the language, that is, such that
 $v(\psi)\neq \nats$ for every $\psi\in L_{\dcoimpscript{}{}\bot}(P)$, trivially respects the interaction rule $\frac{p}{\;\;\dcoimpscript{\,p}{\bot}\;}$,
 and corresponds to the valuation over $\TWO_{\dcoimpscript{\,}{}\bot}$ that sends every sentential variable to~$0$.
If instead $v(\psi)=\nats$ for some $\psi\in L_{\dcoimpscript{\,}{}\bot}(P)$, then $v(\dcoimp{\,\psi}{\bot})=v(\psi)\cap \overline{v(\bot)}=\nats$ implies $v(\bot)=\varnothing$.
Hence,  $v$ is also a valuation over $\TWO^\omega_{\dcoimpscript{}{}\bot}$.
%
%
\hfill$\triangle$
\end{example}

It is worth noting that in some cases the disjoint fibring of the logic of some Boolean connectives with $\bot$ admits a semantics that is simpler than the $4$-valued \nmatrix{} obtained above.

\begin{proposition}\label{tpbot} 
If $\bot\notin\Sigma$ and every connective in $\Sigma$ is truth-preserving then 
 $\cB_{\Sigma}\,\bullet\, \cB_\bot$ has a $4$-valued characteristic logical matrix.
\end{proposition}
\begin{proof}
We know from Prop.~\ref{withbot} that  $\cB_{\Sigma}\,\bullet\, \cB_\bot$ is characterized by the $4$-valued \nmatrix{} $\TWO_\Sigma^2\star\TWO_\bot$. We will show that $\cB_{\Sigma}\,\bullet\, \cB_\bot$ is equivalently characterized by the $4$-valued matrix $\Mt^4_{\Sigma\cup\{\bot\}}:=\tuple{\{0,1\}^2,\{(1,1)\},\widehat{\cdot}_4\,}$ where $\widehat{\conn}_4:=\widetilde{\conn}_2$ for $\conn\in\Sigma$ {(we take $\widetilde{\conn}_2$ as the interpretation of $\conn$ in $\TWO^2_\Sigma$)}
and $\widehat{\bot}_4:=(1,0)$. The argument we use here is a specialization of the one used in the proof of  Prop.~\ref{withtop}.

Let $\cB_{\Sigma}\,\bullet\, \cB_\bot=\cB_{\Sigma}\,\bullet\, \bbbot_{\bot}:=\tuple{\Sigma\cup\{\bot\},\der_{\bbbot}}$, and recall from Prop.~\ref{withtop} that $\cB_{\Sigma}\,\bullet\, \bbbotop_\bot$ is characterized by the \nmatrix{} $\TWO_\Sigma\star\Mt^{\bbbotop}_\bot$.  Let $\tuple{\Sigma\cup\{\bot\},\der_{\bbbotop}}$ refer to $\cB_{\Sigma}\,\bullet\, \bbbotop_{\bot}$. 
%
%
For $\Gamma\cup\{\varphi\}\subseteq L_{\Sigma\cup\{\bot\}}(P)$, note that 
$\Gamma\der_{\bbbot}\varphi$ if and only if $\Gamma\der_{\bbbotop}\varphi$ or {$\Gamma\der_{\bbbotop}\bot$ (as $\cB_\bot$ is axiomatized by the single rule $\frac{\bot}{p}$)}. 
Soundness follows by observing that $\Val_P(\Mt^4_{\Sigma\cup\{\bot\}})=\{v\in\Val_{P}(\TWO^2_\Sigma\star\TWO_\bot):v(\bot)=(1,0)\}$. 

Further, note that $\Val_{P}(\TWO^2_\Sigma\star\TWO_\bot)=\Val_{P}(\TWO^2_\Sigma\star\Mt^{\bbbot}_\bot)=\{v\in\Val_{P}(\TWO^2_\Sigma\star\Mt^{\bbbotop}_\bot):v(\bot)\neq(1,1)\}$, and therefore $\Val_P(\Mt^4_{\Sigma\cup\{\bot\}})=\{v\in\Val_{P}(\TWO^2_\Sigma\star\Mt^{\bbbotop}_\bot):v(\bot)=(1,0)\}=\{v'\circ\skel_\Sigma:v'\in\Val_{P\cup X_\Sigma}(\TWO^2_\Sigma)\mbox{ and }v'(x_\bot)=(1,0)\}=\{(v'_1\circ\skel_\Sigma,v'_2\circ\skel_\Sigma):v'_1,v'_2\in\Val_{P\cup X_\Sigma}(\TWO_\Sigma)\mbox{ and }v'_1(x_\bot)=1\mbox{ and }v'_2(x_\bot)= 0 \}=\{(v_1,v_2):v_1,v_2\in\Val_{P\cup X_\Sigma}(\TWO_\Sigma\star\Mt^{\bbbotop}_\bot)\mbox{ and }v_1(\bot)= 1\mbox{ and } v_2(\bot)= 0 \}$.

As for completeness, if $\Gamma\not\der_{\bbbot}\varphi$ then $\Gamma\not\der_{\bbbotop}\varphi$ and $\Gamma\not\der_{\bbbotop}\bot$. 
Hence, there exist valuations $w_1,w_2\in\Val_P(\TWO_\Sigma\star\Mt^{\bbbotop}_\bot)$ such that $w_1(\Gamma)=w_2(\Gamma)\subseteq\{1\}$ and 
$w_1(\varphi)=w_2(\bot)=0$. Additionally, note that, as the connectives in $\Sigma$ are all assumed to be truth-preserving, then
$v_\curlyvee'\in\Val_{P\cup X_\Sigma}(\TWO_\Sigma)$ where $v_\curlyvee'(\psi)=1$ for every $\psi\in L_\Sigma(P\cup X_\Sigma)$. 
Consider $v_\curlyvee=v_\curlyvee'\circ\skel_\Sigma\in \Val_P(\TWO_\Sigma\star\Mt^{\bbbotop}_\bot)$.

Now, for each $\alpha\in L_{\Sigma}(P)$, let $$v(\alpha):=
\begin{cases}
 (v_\curlyvee(\alpha),w_1(\alpha)),& \mbox{ if }w_1(\bot)=0\\
 (w_1(\alpha),w_2(\alpha)),& \mbox{ if }w_1(\bot)=1
\end{cases}.
$$
In either case, $v\in\Val_P(\TWO_\Sigma^2\star\TWO_\bot)$, $v(\Gamma)\subseteq\{(1,1)\}$, $v(\varphi)\neq(1,1)$ and $v(\bot)=(1,0)$. Precisely because $v(\bot)=(1,0)$, we conclude that $v\in\Val_P(\Mt^4_{\Sigma\cup\{\bot\}})$.
%
%
 \end{proof}

\begin{example}[Implication and bottom]\label{ex:implicabot}
\em
Recall implication, $\imp{}{}$, from Sec.~\ref{classical}, and observe that it is a truth-preserving connective.
When fibring $\cB_{\impscript{}{}}$ and $\cB_{\bot}$, in view of truth-preservation, by Prop.~\ref{tpbot},
we know that $\cB_{\impscript{}{}}\bullet \cB_\bot$ is characterized by the $4$-valued matrix $\Mt^4_{\impscript{}{}\bot}:=\tuple{\{0,1\}^2,\{(1,1)\},\widehat{\cdot}_4}$ where $\widehat{\imp{}{}}_4:=\widetilde{\smash{\imp{}{}}}_2$ (as in $\TWO^2_\impscript{}{}$) and $\widehat{\bot}_4:=(1,0)$, that is:
%
%
%
%
%
\begin{center}
    \begin{minipage}{0.55\textwidth}
        \centering
  \begin{tabular}{c | c c c c}
    $\widehat{\imp{}{}}_4$ & $(0,0)$ & $(0,1)$ & $(1,0)$ & $(1,1)$\\ 
    \hline
    $(0,0)$ & $(1,1)$ & $(1,1)$  & $(1,1)$ & $(1,1)$ \\  
     $(0,1)$  &$(1,0)$ & $(1,1)$  & $(1,0)$ & $(1,1)$ \\  
   $(1,0)$ & $(0,1)$ & $(0,1)$ & $(1,1)$ & $(1,1)$\\
   $(1,1)$ & $(0,0)$ & $(0,1)$ & $(1,0)$ & $(1,1)$
     \end{tabular}
    \end{minipage}
        \begin{minipage}{0.25\textwidth}
        \centering
          \begin{tabular}{c c c}
    $\widehat{\bot}_4$   \\ 
    \hline
    $(1,0)$\\ \\ \\
    \mbox{}
  \end{tabular}
    \end{minipage}%

\end{center}
Further, in $\cB_{\impscript{}{}}\bullet\cB_\bot$ we have $\not\der\bot{\imp{}{}} p$, and so this logic is strictly weaker than~$\cB_{\impscript{}{}\bot}$.
A complete calculus for $\cB_{\impscript{}{}\bot}$ may be obtained by simply adding to the calculus of $\cB_{\impscript{}{}}$ the single interaction axiom:
	\begin{center}
\scalebox{1.2}{
\begin{tabular}{ccccccc}
$\frac{}{\;\;\bot\impscript{}{} p\;\;}$\\[.4mm]
\end{tabular}
}
\end{center}
The usual rule for $\bot$ is easily derivable. Completeness of the resulting calculus may be easily confirmed using Lemma~\ref{addaxioms}(b). 
Indeed, note that there are two kinds of valuations over $\Mt^4_{\impscript{}{}\bot}$ that respect the axiom $\imp{\bot}{\,p}\ $:
either $v(\bot)=(0,0)$, in which case it is also a valuation over $\TWO^2_{\impscript{}{}\bot}$, or $v(\bot)=(1,0)$ (resp.\ $v(\bot)=(0,1)$), in which case the only possible values for the other formulas are $(1,0)$ or $(1,1)$ (resp.\ $(0,1)$ or $(1,1)$). So, $\pi_2(v)$ (resp.\ $\pi_1(v)$) is a valuation over $\TWO_{\impscript{}{}\bot}$ satisfying the same formulas as~$v$.
\end{example}

The semantic characterizations provided by Prop.~\ref{withbot} and Prop.~\ref{tpbot} may still be further improved in the very particular, and perhaps surprising, case where all the Boolean connectives in~$\Sigma$ are expressible as derived connectives with the sole use of bi-implication.


\begin{proposition} \label{propequiv}
If $\bot\notin\Sigma$ and $\cC_{2}^{\Sigma}\subseteq  \cC_{2}^{\biimpscript{}{}}$ then 
$\cB_{\Sigma}\,\bullet\, \cB_\bot=\cB_{\Sigma\cup\{\bot\}}$.
\end{proposition}

\begin{proof}
First, we prove that $\cB_{\smash{{\biimpscript{}{}}}{\top}}\,\bullet\, \cB_\bot=\cB_{\smash{{\biimpscript{}{}}}{\top}\bot}$. Since $\biimp{}{}$ and ${\top}$ are truth-preser\-ving, we know from Prop.~\ref{tpbot} that $\cB_{\smash{{\biimpscript{}{}}}{\top}}\,\bullet\, \cB_\bot$ is characterized by the $4$-valued matrix $\Mt^4_{\smash{{\biimpscript{}{}}}{\top}\bot}:=\tuple{\{0,1\}^2,\{(1,1)\},\widehat{\cdot}_{4}\,}$ where 
$\widehat{\smash{\biimp{}{}}}_4:=\widetilde{\smash{\biimp{}{}}}_2$ and $\widehat{{\top}}_4:=\widetilde{{\top}}_2=(1,1)$
(thus extending $\TWO^2_{\smash{{\biimpscript{}{}}}{\top}}$) and $\widehat{\bot}_4:=(1,0)$. 
We will show that $\cB_{\smash{{\biimpscript{}{}}}{\top}}\,\bullet\, \cB_\bot$ is equivalently characterized by $\TWO_{\smash{{\biimpscript{}{}}}{\top}\bot}$.

Consider the bijection $h:\{0,1\}^2\longrightarrow \{0,1\}^2$ such that $h(1,a)=(a,a)$ and $h(0,a)=(1-a,a)$ for $a\in\{0,1\}$. It is straightforward to check that $h$ establishes an isomorphism between $\Mt^4_{\smash{{\biimpscript{}{}}}{\top}\bot}$ and $\TWO^2_{\smash{{\biimpscript{}{}}}{\top}\bot}$. 
Indeed, first note that $h({\top})=h(1,1)=(1,1)=\widetilde{{\top}}_2$, and 
$h(\widehat{\bot}_4)=h(1,0)=(0,0)=\widetilde{\bot}_2$. To see that $h((a_1,b_1)\widehat{\smash{\biimp{}{}}}_4(a_2,b_2))=h(a_1,b_1)\widetilde{\smash{\biimp{}{}}}_2 h(a_2,b_2)$ note that  $\widehat{\smash{\biimp{}{}}}_4=\widetilde{\smash{\biimp{}{}}}_2$ is commutative and analyze the possible cases: (i) $h((1,1)\widetilde{\smash{\biimp{}{}}}_2(a,b))=h(a,b)=(1,1)\widetilde{\smash{\biimp{}{}}}_2 h(a,b)=h(1,1)\widetilde{\smash{\biimp{}{}}}_2 h(a,b)$; (ii) $h((a,b)\widetilde{\smash{\biimp{}{}}}_2(a,b))=h(1,1)=h(a,b)\widetilde{\smash{\biimp{}{}}}_2 \linebreak h(a,b)$; and (iii) if $(a_1,b_1), (a_2,b_2)$ are two distinct undesignated values and $(a_3,b_3)$ is the other undesignated value, then $h((a_1,b_1)\widetilde{\smash{\biimp{}{}}}_2(a_2,b_2))=h(a_3,b_3)
=h(a_1,b_1)\widetilde{\smash{\biimp{}{}}}_2 h(a_2,b_2)$. This shows that $\cB_{\smash{{\biimpscript{}{}}}{\top}}\,\bullet\, \cB_\bot$ is equivalently characterized by~$\TWO^2_{\smash{{\biimpscript{}{}}}{\top}\bot}$, and thus also by $\TWO_{\smash{{\biimpscript{}{}}}{\top}\bot}$.

{
Finally, consider $\bt_1:L_{\Sigma}(P)\longrightarrow L_{{\biimpscript{}{}}{\top}}(P)$ such that 
$\TWO_\Sigma=\TWO_{{\biimpscript{}{}}{\top}}^{\bt_1}$ and $\bt_2:=\mathsf{id}_\bot$.
We are thus under applicability conditions of Prop.~\ref{liftspecial}$\mathbf{(c)}$ and Prop.~\ref{lift},
and
from  $\TWO_{\Sigma\cup\{\bot\}}=\TWO^{\bt_1\cup \bt_2}_{{\biimpscript{}{}}{\top}\bot}$ we
 conclude that
  $\fib{\cB_{\Sigma}\!}{\cB_{\bot}}=\cB_{\Sigma\cup\{\bot\}}$.}
\end{proof}

The next example illustrates a rather special ---and perhaps unexpected--- situation: the Boolean logic of bi-implication and $\bot$ coincides with the fibring of the corresponding one-connective fragments. This fact applies also if we replace bi-implication with the connective $\mathsf{+^3}$ which is expressible using $\biimp{}{}$ by setting $\lambda p_1p_2p_3.\,\biimp{p_1}{(\biimp{p_2}{p_3})}$. These results are to be contrasted, in the light of Prop.~\ref{badbottom} below, with the fibring of $\bot$ with any connective in the list [\texttt{L1}]: $T^{n+2}_{n+1}$ (for $n\in\nats$), $T^{n+4}_2$ (for $n\in\nats$), $\neg$, $\imp{}{}$, $\dcoimp{}{}$, $+$, $\textsc{if}$,  $\lambda p_1p_2p_3.\,p_1\lor(p_2\land p_3)$, $\lambda p_1p_2p_3.\,p_1\lor(p_2+p_3)$, $\lambda p_1p_2p_3.\,p_1\land(p_2\lor p_3)$, $\lambda p_1p_2p_3.\,p_1\land(\imp{p_2}{p_3})$. Note that $[\texttt{L1}]$ contains all the connectives in [\texttt{L0}] except $\biimp{}{}$ and $\mathsf{+^3}$.

\begin{example}[Bi-implication and bottom]\label{equivbot}
\em
We consider combining~$\cB_{\biimpscript{}{}}$ and~$\cB_\bot$. 
It follows from Prop.~\ref{propequiv}
that $\cB_{\biimp{}{}}\bullet\cB_\bot=\cB_{\biimp{}{}\bot}$. Thus, the fibred logic is $2$-valued and is characterized by the matrix $\TWO_{\biimp{}{}\bot}$.

A complete calculus for $\cB_{\biimp{}{}\bot}$ may be obtained by simply merging calculi for the components (a calculus for $\cB_{\biimp{}{}}$ may be found in~\cite[p. 332]{Rautenberg1981}).
%
\hfill$\triangle$
\end{example}

The following example provides further illustration on Prop.~\ref{prop:failingmaya}, and highlights the role of the condition concerning the nullarity of bottom in Prop.~\ref{propequiv}:

\begin{example}[Bi-implication and $1$-place bottom]\label{equivbotun}
\em
{Let the connective $\bot^{\!1}$ be a $1$-place bottom-like connective. }
This time we consider combining $\cB_{\biimp{}{}}$ with the logic~$\cB_{\bot^{\!1}}$ of~$\bot^{\!1}$, characterized by the Boolean matrix $\TWO_{\bot^{\!1}}$, which is known to be saturated by Prop.~\ref{satveysig}.
As~$\cB_{\biimpscript_{}{}}$ is not saturated, we consider instead $\TWO^\omega_{\biimpscript{}{}}$.
It follows from the results mentioned in Rem.~\ref{sem-fib} that 
$\cB_{\biimpscript{}{}}\bullet \cB_{\bot^{\!1}}$ is characterized by the non-de\-numerable \nmatrix{} $\TWO^\omega_{\biimpscript{}{}}\star \TWO_{\bot^{\!1}}=\tuple{2^\nats,\{\nats\},\widetilde{\cdot}_\star}$ where 
$X\widetilde{\smash{\biimp{}{}}}_\star Y:= (\overline{X}\cup Y)\cap (\overline{Y}\cup X)$ and 
$\widetilde{\bot^{\!1}}_\star (X):=\{Y:Y\neq \nats\}$.
 
As $\bot^{\!1}$ is a non-top-like Boolean connective distinct from the $0$-place connective~$\bot$, and~$\biimp{}{}$ is very significant,
by Prop.~\ref{prop:failingmaya} we know that  
$\cB_{\biimpscript{}{}}\bullet \cB_{\bot^{\!1}}$ is not characterized by a finite matrix.
Furthermore, we claim that $\cB_{\biimpscript{}{}}\bullet \cB_{\bot^{\!1}}$ is not even finitely-\nvalued{}. 
We will show indeed that it is not $k$-determined.
Let $\psi_i:=\bot^{\!1}(p_i)$, for $1\leq i\leq k+1$, and 
$\Gamma:=\{\biimp{(\biimp{\psi_i}{\psi_j})}{\psi_\ell}:i\neq j, i\neq \ell, j\neq \ell, 1\leq i,j,\ell\leq k+1\}$.
We have that $\Gamma\not\der p_{k+2}$.
However, given  
$\sigma:P\longrightarrow \{p_1,\ldots,p_k\}$, 
it follows by the pigeonhole principle that there must be some $i\neq j$ such that
$\psi_i^\sigma=\psi_j^\sigma$, and so $\Gamma^\sigma\der\psi_\ell$.
As $\psi_\ell$ is bottom-like, i.e., $\psi_\ell\der p$, we obtain $\Gamma^\sigma\der \sigma(p_{k+2})$.
Hence, $\cB_{\biimpscript{}{}}\bullet \cB_{\bot^{\!1}}$ is strictly weaker than $\cB_{\biimpscript{}{}\bot^{\!1}}$.

A complete calculus for $\cB_{\biimpscript{}{}\bot^{\!1}}$ may be obtained by simply adding to a calculus for $\cB_{\biimpscript{}{}}\bullet \cB_{\bot^{\!1}}$ the single interaction axiom $\frac{}{\biimpscript{\bot^{\!1}(p)}{\bot^{\!1}(q)}}$. Completeness follows by Lemma~\ref{addrules}(b)
and the fact that 
$A\widetilde{\smash{\biimp{}{}}}_\star B$ takes a designated value if and only if $A=B$, therefore every $\bot^{\!1}$-headed formula must have the same value, and the functions that swap some undesignated points with~$\varnothing$
are isomorphisms between $\TWO^\omega_{\biimpscript{}{}}\star \TWO_{\bot^{\!1}}$ and $\TWO^\omega_{\biimpscript{}{}\bot^{\!1}}$. We abstain from presenting here further details as they are in fact very similar to the argument presented to the same effect in Ex.~\ref{equivbot}. 
\hfill$\triangle$
\end{example}

In contrast to the above, the following example shows that the situation changes if we simultaneously add two $0$-place bottoms, in which case a subclassical logic is obtained. We will consider the connective $\mathsf{+^3}$, but the same argument would apply to the connective $\biimp{}{}$.

\begin{example}[$\mathsf{+^3}$ and two bottoms]\label{3+botS}
\em
Let $\bot_1$ and $\bot_2$ be two $0$-place bottoms.
We consider merging $\cB_{\mathsf{+^3}}$ and the logic of these two bottom-like connectives, $\cB_{\bot_1\bot_2}$, characterized by $\TWO_{\bot_1\bot_2}$.  
Clearly, $\TWO_{\bot_1\bot_2}$ is still saturated in view of Prop.~\ref{satveysig}.
Following the same recipe as in the case of a single bottom in Prop.~\ref{withtop} (but now with $3$-saturation), 
we immediately conclude that
$\cB_{\biimpscript{}{}}\bullet \cB_{\bot_1\bot_2}$
is characterized by the \nmatrix{} $\TWO^3_{\mathsf{+^3}}\star \TWO_{\bot_1\bot_2}$.
Choosing~$v$ over $\TWO^3_{\mathsf{+^3}}\star \TWO_{\bot_1\bot_2}$
 such that $v(p):=(0,1,1)$, 
 $v(\bot_1):=(1,0,0)$ and $v(\bot_2):=(0,0,0)$,
 we have that 
$v(\mathsf{+^3}(p,\bot_1,\bot_2))=(1,1,1)$. Hence, 
 the mixed consequence assertion $\mathsf{+^3}(p,\bot_1,\bot_2)\der{p}$ fails to hold in the fibred logic, and so
  $\cB_{\mathsf{+^3}}\bullet \cB_{\bot_1\bot_2}\neq \cB_{{\mathsf{+^3}}\bot_1\bot_2}$.
  
An axiomatization for $\cB_{\mathsf{+^3}}$ may be found in~\cite[p. 331]{Rautenberg1981}, and $\cB_{\bot_1\bot_2}$ is axiomatized simply by the rules $\frac{\bot_1}{p}$ and $\frac{\bot_2}{p}$.
  A complete calculus for $\cB_{\mathsf{+^3}\bot_1\bot_2}$ may be obtained by just adding to a calculus for $\cB_{\mathsf{+^3}}\bullet \cB_{\bot_1\bot_2}$ the two interaction rules $\frac{\mathsf{+^3}(\bot_1,p,q)}{\mathsf{+^3}(\bot_2,p,q)}$ and $\frac{\mathsf{+^3}(\bot_2,p,q)}{\mathsf{+^3}(\bot_1,p,q)}$. 
  %
\hfill$\triangle$
\end{example}

We see that the Boolean connectives definable by bi-implication still result in a two-valued classical logic when combined with~$\bot$. This can never be the case with other connectives, as we show below. 
We shall prove that the result of adding~$\bot$ to a logic expressing any connective from [\texttt{L1}] (or equivalently a connective from [\texttt{L0}] that does not belong to~$\cC_2^\biimpscript{}{}$) fails to yield the corresponding fragment of classical logic.

\begin{proposition} 
\label{badbottom}
If $\bot\notin\Sigma$ and $\widetilde{\conn}\in  \cC_{2}^{\Sigma}$ for $\conn$ in $ [\texttt{L1}]$
    then $\fib{\cB_{\Sigma}}{\cB_{\bot}}\subsetneq \cB_{\Sigma\cup\{\bot\}}$.
\end{proposition}

\begin{proof}
Consider the list of connectives [\texttt{L2}]: $\ou,T^3_2,\neg,+,\lambda p_1p_2p_3.\,p_1\land(p_2\lor p_3)$. Observe that [\texttt{L2}] is a sublist of [\texttt{L1}].
First, we prove that $\fib{\cB_{\conn}}{\cB_{\bot}}\subsetneq \cB_{\conn\bot}$ for $\conn$ in [\texttt{L2}]. Note that 
$\ou=T^2_1$. In all cases, we shall take advantage of Prop.~\ref{withbot}, which shows that $\cB_{\conn}\bullet\cB_{\bot}$ is characterized by the $4$-valued \nmatrix{} $\TWO^2_\conn\star\TWO_\bot$.

\begin{itemize}
\item[---] If $\conn=\ou$ then $\bot \ou p\der p$ holds classically but fails to hold in $\cB_{\ou}\bullet\cB_{\bot}$, as shown by a valuation $v\in\Val_P(\TWO^2_\ou\star\TWO_\bot)$ with $v(\bot)=(0,1)$ and $v(p)=(1,0)\neq (1,1)$, which is such that $v(\bot \ou p)=(0,1)\widetilde{\ou}{}_2(1,0)=(1,1)$.

\item[---] If $\conn=T^3_2$ then $T^3_2(\bot,p,q)\der p$ holds classically but fails to hold in $\cB_{T^3_2}\bullet\cB_{\bot}$, as shown by a valuation $v\in\Val_P(\TWO^2_{T^3_2}\star\TWO_\bot)$ with $v(\bot)=(0,1)$, $v(p)=(1,0)\neq (1,1)$ and $v(q)=(1,1)$, which is such that $v(T^3_2(\bot,p,q))=
\widetilde{(T^3_2)}{}_2((0,1),(1,0),(1,1))=(1,1)$.

\item[---] If $\conn=\neg$ then $\der \neg\bot$ holds classically but fails to hold in $\cB_{\neg}\bullet\cB_{\bot}$, as shown by a valuation $v\in\Val_P(\TWO^2_{\neg}\star\TWO_\bot)$ with $v(\bot)=(0,1)$, for which necessarily $v(\neg\bot)=
\widetilde{\neg}{}_2(0,1)=(1,0)\neq (1,1)$.

\item[---] If $\conn=+$ then $\bot + p\der p$ holds classically but fails to hold in $\cB_{+}\bullet\cB_{\bot}$, as shown by a valuation $v\in\Val_P(\TWO^2_+\star\TWO_\bot)$ with $v(\bot)=(0,1)$ and $v(p)=(1,0)\neq (1,1)$, for which necessarily $v(\bot + p)=(0,1)\widetilde{+}{}_2(1,0)=(1,1)$.

\item[---] If $\conn=\lambda p_1p_2p_3.\,p_1\land(p_2\lor p_3)$ then $p\land(\bot\lor q)\der q$ holds classically but fails to hold in $\cB_{\conn}\bullet\cB_{\bot}$, as shown by a valuation $v\in\Val_P(\TWO^2_{\conn}\star\TWO_\bot)$ with $v(\bot)=(0,1)$, $v(p)=(1,1)$ and $v(q)=(1,0)\neq(1,1)$, for which necessarily 
$v(p\land(\bot\lor q))=(1,1)\widetilde{\e}{}_2((0,1)\widetilde{\ou}{}_2(1,0))=(1,1)\widetilde{\e}{}_2(1,1)=(1,1)$.
\end{itemize}

\noindent
As it is clear that $\fib{\cB_{\conn\!}}{\cB_{\bot}}\subseteq \cB_{\conn\bot}$, in all cases considered above, we conclude that $\fib{\cB_{\conn\!}}{\cB_{\bot}}\subsetneq \cB_{\conn\bot}$, for~$\conn$ a connective from the restricted list [\texttt{L2}]. 

We now note that each of the other connectives in [\texttt{L1}] expresses some connective from [\texttt{L2}] (actually, in all cases, either~$\ou$ or $\bowtie\,:=\lambda p_1p_2p_3.\,p_1\land(p_2\lor p_3)$ may be seen to be a derived connective).

\begin{itemize}
\item[---] If $\conn=T^{n+2}_{n+1}$ with $n\geq 2$ then $\bowtie$ is expressed by $\lambda p_1p_2p_3.\,T_{n+1}^{n+2}(\overline{p_1}^n,p_2,p_3)$.
\item[---] If $\conn=T^{n+4}_{2}$ with $n\in\nats$ then $\ou$ is expressed by $\lambda p_1p_2.\,T_{2}^{n+4}(\overline{p_1}^2,\overline{p_2}^{n+2})$.
\item[---] If $\conn=\,\imp{}{}$ then $\ou$ is expressed by $\lambda p_1p_2.\,\imp{(\imp{p_1}{p_2})}{p_2}$.
\item[---] If $\conn=\,\dcoimp{}{}$ then $\bowtie$ is expressed by $\lambda p_1p_2p_3.\,\dcoimp{p_1}{(\dcoimp{(\dcoimp{p_1}{p_2})}{p_3})}$.
\item[---] If $\conn=\textsc{if}$ then $\ou$ is expressed by $\lambda p_1p_2.\,\textsc{if}(p_1,p_1,p_2)$.
\item[---] If $\conn=\lambda p_1p_2p_3.\,p_1\lor(p_2\land p_3)$ then $\ou$ is expressed by $\lambda p_1p_2.\,p_1\lor(p_2\land p_2)$.
\item[---] If $\conn=\lambda p_1p_2p_3.\,p_1\lor(p_2 + p_3)$ then $\ou$ is expressed by $\lambda p_1p_2.\,p_1\lor(p_1+ p_2)$.
\item[---] If $\conn=\lambda p_1p_2p_3.\,p_1\e(\imp{p_2}{p_3})$ then $\bowtie$ is expressed by $\lambda p_1p_2p_3.\,p_1\e(\imp{(p_1\e(\imp{p_2}{p_3}))}{p_3})$.
\end{itemize} 

\noindent
This means that if $\widetilde{\conn}\in  \cC_{2}^{\Sigma}$ for some connective in the list [\texttt{L1}] given in the statement, then the same is true also for one of the five connectives in the smaller list [\texttt{L2}]. Thus, we can assume without loss of generality that $\conn$ is in $\text{[\texttt{L2}]}$. 
{Let $\bt_\conn$ be such that $\TWO_\conn=\TWO^{\bt_\conn}_\Sigma$. 
We are thus under the applicability conditions of Prop.~\ref{liftspecial}$\mathbf{(c)}$ and Prop.~\ref{lift} with $\bt=\bt_\conn\cup \mathsf{id}_{\bot}$.
Hence, from $\cB_{\conn}\,\bullet\, \cB_\bot\neq \cB_{\TWO_{\conn \bot}}$ and $\TWO_{\conn\bot}=\TWO^\bt_{\Sigma\cup\{\bot\}}$, we conclude that 
$\cB_{\Sigma}\,\bullet\, \cB_\bot\neq 
\cB_{\TWO_{\Sigma \bot}}$.
%
%
}
\end{proof}


%
%
%
%
%
%

The following example illustrates, together with Ex.~\ref{coimpbot} and Ex.~\ref{ex:implicabot}, the variety of behaviors that may arise from adding a $0$-place bottom to the logic of a single classical connective, as described by Prop.~\ref{badbottom}:

\begin{example}[Negation and bottom]\label{negbot}
\em
We now consider fibring the logics $\cB_\neg$ and $\cB_\bot$ of classical negation and bottom. 
Recall from Ex.~\ref{twonegs} that $\TWO_\neg$ is not saturated but we can consider instead the $3$-valued saturated matrix $\Mt^3_\neg$. Note also that $\Mt^{\bbbot}_\bot=\TWO_\bot$, and from Prop.~\ref{withtop}  it follows that 
$\cB_\neg\bullet\cB_\bot$ is characterized by the $3$-valued \nmatrix{} $\Mt^3_\neg\star \TWO_\bot=
 \tuple{\{0,\frac{1}{2},1\},\{1\},\widetilde{\cdot}_\star}$
with $\widetilde{\neg}_\star:=\widetilde{\neg}^3$ and $\widetilde{\bot}_\star:=\{0,\frac{1}{2}\}$.

	To see that $\cB_{\neg}\bullet\cB_\bot$ is not deterministically many-valued we point out the fact that $\neg \bot \der \neg \bot$ but $\not\der \neg \bot$, which implies the failure of the cancellation property. It follows, of course, that $\cB_{\neg}\bullet\cB_\bot$ is strictly weaker than $\cB_{\neg\bot}$.
	A~complete calculus for $\cB_{\neg\bot}$ may be obtained by adding to the calculus of $\cB_{\neg}$ a single interaction axiom, namely:
	\begin{center}
\scalebox{1.2}{
\begin{tabular}{ccccccc}
$\frac{}{\;\; \neg\bot\;\;}$
\end{tabular}
}
\end{center}
\noindent
{Completeness of the resulting calculus may easily be confirmed using Lemma~\ref{addaxioms}(b). 
Indeed, we know from Prop.~\ref{withtop} that $\TWO_\neg^2\star \TWO_\bot$ also defines the same logic.
Clearly, any~$v$ that validates the above interaction axiom $\neg\bot$ is such that $v(\bot)=(0,0)$, therefore $v$ is actually a valuation over $\TWO_{\neg\bot}^2$. 
}
\hfill$\triangle$
%
%
\end{example}

\section{Putting it all together}
\label{sec:summary}

\subsection{Characterizing the Boolean combinations}

Building on Prop.~\ref{topisnice}, Prop.~\ref{notverydifficult}, Prop.~\ref{prop:failingmaya}, Prop.~\ref{propequiv} and Prop.~\ref{badbottom}, from the previous subsections, we are finally able to identify, in the next theorem, the precise conditions for the recovery of a fragment of classical logic by fibring disjoint classical components. The facts about Boolean clones highlighted in Rem.~\ref{rem:clones} turn out to be essential in proving the result, which takes indeed full advantage of the fact that every very significant connective expresses some connective in $ [\texttt{L0}]$.

 




%
%
%
%


\begin{theorem} 
\label{lem:colapse}
If the signatures $\Sigma_1$ and $\Sigma_2$ are disjoint then $\fib{\cB_{\Sigma_1\!}}{\cB_{\Sigma_2}}={\cB_{\Sigma_1\cup\Sigma_2}}$ if and only if  \underline{either}: \begin{itemize}
 \item[$\mathbf{(a)}$]  $\cC_2^{\Sigma_i}\subseteq\cC_2^\top$ for some $i\in\{1,2\}$, 
 \underline{or}
 \item[$\mathbf{(b)}$]  $\cC_2^{\Sigma_1},\cC_2^{\Sigma_2}\subseteq\cC_{2}^{\land\top\bot}$, 
 \underline{or}
 \item[$\mathbf{(c)}$] $\cC_2^{\Sigma_i}\subseteq\cC_2^{\biimpscript{}{}}$ and $\bot\in\Sigma_j$ is the only non-top like connective in $\Sigma_j$, for $i\neq j$ with $i,j\in\{1,2\}$.
 
\end{itemize}
\end{theorem}

\begin{proof}
If {\bf(a)} is the case, then $\fib{\cB_{\Sigma_1\!}}{\cB_{\Sigma_2}}={\cB_{\Sigma_1\cup\Sigma_2}}$ follows from Prop.~\ref{topisnice}.
If {\bf(b)} is the case, then $\fib{\cB_{\Sigma_1\!}}{\cB_{\Sigma_2}}={\cB_{\Sigma_1\cup\Sigma_2}}$ follows from Prop.~\ref{notverydifficult}.
If {\bf(c)} is the case, then assume without loss of generality that $i=1$ and $j=2$. First observe that $\cB_{\Sigma_2}=\cB_{\Sigma_2\setminus\{\bot\}}\bullet\cB_\bot$, and also $\cB_{(\Sigma_1\cup\Sigma_2)\setminus\{\bot\}}=\cB_{\Sigma_1}\bullet\cB_{\Sigma_2\setminus\{\bot\}}$, both facts being justified by Prop.~\ref{topisnice} as 
$\cC_2^{\Sigma_2\setminus\{\bot\}}\subseteq\cC_2^\top$.
 Thus, $\fib{\cB_{\Sigma_1\!}}{\cB_{\Sigma_2}}={\cB_{\Sigma_1\!}}\bullet({\fib{\cB_{\Sigma_2\setminus\{\bot\}\!}}{\cB_\bot}})=(\cB_{\Sigma_1\!}\bullet \cB_{\Sigma_2\setminus\{\bot\}\!})\bullet\cB_\bot=\cB_{(\Sigma_1\cup\Sigma_2)\setminus\{\bot\}}\bullet\cB_\bot={\cB_{\Sigma_1\cup\Sigma_2}}$, this last step being a consequence of Prop.~\ref{propequiv} since $\cC_2^{\top}\subseteq\cC_2^{\biimpscript{}{}}$, and therefore $\cC_2^{(\Sigma_1\cup\Sigma_2)\setminus\{\bot\}}\subseteq\cC_2^{\biimpscript{}{}}$.\smallskip

Conversely, assume that $\fib{\cB_{\Sigma_1\!}}{\cB_{\Sigma_2}}={\cB_{\Sigma_1\cup\Sigma_2}}$ and that neither {\bf(a)} nor {\bf(b)} are the case. This means that one of the logics, say $\cB_{\Sigma_1}$, expresses a very significant connective, while the other logic, $\cB_{\Sigma_2}$, expresses some non-top-like connective~$\conn$. From Prop.~\ref{prop:failingmaya} and Prop.~\ref{halflift}
it follows that $\conn=\bot$. 
Thus, it follows from Prop.~\ref{badbottom} (and Rem.~\ref{rem:clones}) and Prop.~\ref{halflift} that $\cC_2^{\Sigma_1}\subseteq\cC_2^{\biimpscript{}{}}$. However, Ex.~\ref{3+botS} and Prop.~\ref{halflift} show that there cannot be two syntactically distinct $0$-place bottom connectives, since $\lambda p_1p_2p_3.\,p_1+p_2+p_3$ is expressed by every
very significant connective in $\cC_2^{\biimpscript{}{}}$. Hence, we conclude that $\bot\in\Sigma_2$ is the only non-top-like connective in~$\Sigma_2$.
\end{proof}

The following result formulates the precise conditions under which full classical logic may be recovered by fibring disjoint fragments of it. Again we take advantage of Post's lattice, highlighted in Rem.~\ref{rem:clones}, namely using the identification of the Boolean clones which are maximal with respect to $\top$.

\begin{corollary} 
\label{fragments-that-work}
Let $\Sigma_1$ and $\Sigma_2$ be disjoint signatures such that $\cC_2^{\Sigma_1},\cC_2^{\Sigma_2}\neq\cC_{2}$ but $\cC_2^{\Sigma_1\cup\Sigma_2}=\cC_{2}$. 
%
Then, $\fib{\cB_{\Sigma_1\!}}{\cB_{\Sigma_2}}={\cB_{\Sigma_1\cup\Sigma_2}}$ if and only if $\cC_{2}^{\Sigma_i}\in \{\mathcal{D},\mathcal{T}_0^\infty\}\cup \{\mathcal{T}_0^{n+1}:n\in \nats\}$ and $\cC_{2}^{\Sigma_j}= \mathcal{UP}_1$, {for $i\neq j$ with $i,j\in \{1,2\}$}.
\end{corollary}
\begin{proof}
If $\fib{\cB_{\Sigma_1\!}}{\cB_{\Sigma_2}}={\cB_{\Sigma_1\cup\Sigma_2}}$, then one of the cases into which Thm.~\ref{lem:colapse} is divided must apply. However, cases~({\bf b}) and~({\bf c}) are not possible because in those cases we would have $\cC_2^{\Sigma_1\cup\Sigma_2}\neq\cC_{2}$. Thus, ({\bf a}) must be the case, and $\cC_2^{\Sigma_i}\subseteq\cC_2^\top$, and actually $\cC_2^{\Sigma_i}=\cC_2^\top=\mathcal{UP}_1$ simply because the other possibility would be $\cC_2^{\Sigma_i}=\cC_2^\varnothing$ (the only clone properly contained in $\cC_2^\top$ is $\cC_2^\varnothing$).
Hence, $\cC_{2}^{\Sigma_j}\in \{\mathcal{D},\mathcal{T}_0^\infty\}\cup \{\mathcal{T}_0^{n+1}:n\in \nats\}$, as these are precisely the clones which are maximal with respect to $\mathcal{UP}_1$ (see Rem.~\ref{rem:clones}).
\end{proof}


\subsection{Summing it up}

In Section~\ref{sec:merging} we have analyzed several examples of combinations of classical connectives produced through fibring (namely, by merging the corresponding axiomatizations), including their characterizations through (logical) (N)matrices, as well as the interaction principles needed for the corresponding fragment of classical logic to be recovered. 
It is worth taking a more abstract look at these examples and the results that structure them.

A first batch of examples that was considered concerned the cohabitation, in the same logic, of two copies of the same Boolean connective. 
Already there one can find all sorts of interesting phenomena arising.  As shown, the addition (through fibring) to the logic of classical conjunction of another copy of classical conjunction, with the same behavior, 
makes these connectives collapse into one another ($\cB_\e\bullet\cB_\&=\cB_{\e\&}$, Ex.~\ref{ex:conj}). 
On the other hand, the analogous collapse does not occur if we combine, say, two copies of negation ($\cB_\neg\bullet\cB_\sim\neq\cB_{\neg \sim}$, Ex.~\ref{twonegs}), or two copies of disjunction ($\cB_\ou\bullet\cB_{||}\neq\cB_{\ou ||}$, Ex.~\ref{ex:disj}). As we have pointed out, the fibring of two copies of the logic of classical negation does not have a finite-valued characterization, yet is $5$-Nvalued, and the fibring of two copies of the logic of classical disjunction does not even have a finite-valued non-deterministic semantics.

Another batch of examples we have entertained involved the combination of two distinct Boolean connectives. Again, if such combination is produced via fibring, aiming at a common minimal conservative extension of the logics of the connectives given as input, several different phenomena may be observed.  In most interesting cases (such as conjunction plus disjunction $\cB_\e\bullet\cB_\ou\neq\cB_{\e\ou}$, Ex.~\ref{ex:conjdisj}) the combined logic turns out to be subclassical and not characterizable by a finite-valued Nmatrix, and this is also the case in situations (such as disjunction plus negation: $\cB_\ou\bullet\cB_\neg\neq\cB_{\ou\neg}$, Ex.~\ref{ex:disjneg}) in which one could have expected the resulting logic to be functionally complete. However, there are cases (such as coimplication plus top: $\cB_{\dcoimpscript{}{}}\bullet\cB_\top=\cB_{\dcoimpscript{}{} \top}$, Ex.~\ref{coimptop}) in which one actually does obtain full classical logic without the need to impose any sort of additional interaction principles involving the two connectives being combined.

A~particular class of examples that deserved separate attention above involved the combination of the logic of some standard classical connectives with the logic of bottom-like connectives. To a bystander unaware of the results in the present paper, the semantic behaviour observed in this last batch of examples might seem erratic. For instance, while combining the logics of negation and of bottom gives rise to a 3-Nvalued logic ($\cB_\neg\bullet\cB_\bot\neq\cB_{\neg\bot}$, Ex.~\ref{negbot}), and combining the logics of complication and of bottom gives rise to a $4$-Nvalued logic ($\cB_{\dcoimpscript{}{}}\bullet\cB_\bot\neq\cB_{\dcoimpscript{}{}\bot}$, Ex.~\ref{coimpbot}), adding a bottom to the logic of implication results in a deterministically $4$-valued logic ($\cB_{\impscript{}{}}\bullet\cB_\bot\neq\cB_{\impscript{}{}\bot}$, Ex.~\ref{ex:implicabot}). Other curious examples include the addition of a bottom to the logic of bi-implication, which outputs the corresponding fragment of classical logic without the addition of interaction principles ($\cB_{\biimpscript{}{}}\bullet\cB_\bot=\cB_{\biimpscript{}{}\bot}$, Ex.~\ref{equivbot}), and the alternative addition of a 1-place bottom-like connective to the same logic of bi-implication ($\cB_{\biimpscript{}{}}\bullet \cB_{\bot^{\!1}}\neq\cB_{\biimpscript{}{}\bot^{\!1}}$, Ex.~\ref{equivbotun}), which results subclassical, instead.  We have also considered an example in which the logic of a ternary odd-counter (a ternary connective that is true iff exactly one or three of its arguments is true) is fibred with the logic containing two copies of the classical bottom, and the resulting logic turned out to be $8$-Nvalued ($\cB_{\mathsf{+^3}}\bullet \cB_{\bot_1\bot_2}\neq\cB_{\mathsf{+^3}\bot_1\bot_2}$  , Ex.~\ref{3+botS}).

The above mentioned seemingly capricious diet of examples was employed both in motivating and in illustrating the results obtained in the present paper.   Substantially advancing beyond the results of the investigation done at our earlier paper \cite{ccal:jmar:smar:merging}, we have in the preceding subsection at last identified, in Thm.~\ref{lem:colapse} and Cor.~\ref{fragments-that-work}, the precise conditions for the recovery of a fragment of classical logic (for any arbitrary signature, with a 2-valued interpretation in terms of logical matrices) through the fibring of disjoint Boolean components.  It is worth mentioning, nonetheless, that some intermediate results obtained while establishing the foundations for these main results have helped identifying some sufficient conditions for a logic (not) to be finitely-valued (Prop.~\ref{prop:failingmaya}), and in several cases we directly showed that our illustrations had (or did not have) a non-deterministic finite-valued characterization.

\section{What lies ahead}%
\label{sec:conclusion}

\noindent 
In this paper we have fully uncovered the conditions under which merging the Hilbert calculi of disjoint fragments of classical logic still leads to a fragment of classical logic, or potentially to full classical logic, without the need to introduce further inference rules regulating the interaction between the connectives from each of the fragments. It comes as no surprise that this is an extremely rare event, but there are a few non-trivial and perhaps unexpected exceptions, fully identified at Thm.~\ref{lem:colapse} and Cor.~\ref{fragments-that-work}. The proofs of these results, which we believe to be entirely novel, rely in an essential way on the ingenious classification of two-valued clones by Post~\cite{Post41}. 
Analogous results for fragments of other important logics are thus expected to be far from straightforward. 
It is worth noting that as a byproduct of Prop.~\ref{notverydifficult} and Prop.~\ref{prop:failingmaya}, we have also fully characterized the circumstances under which collapses of classical connectives are produced via fibring, namely, when we are dealing with two copies of a Boolean connective that is not very significant.

Some of the results and the general techniques used in this paper are, nonetheless, applicable well beyond classical logic. Overall, the present investigation may be seen as an application of the recent semantic characterization of disjoint fibring in~\cite{smar:ccal:17}, which uses in a fundamental way the advantages of the non-deterministic environment permitted by \nmatrices{}. 
The myriad of interesting subclassical logics 
that are obtained in all the cases in which the combination of classical fragments fails to be classical, as illustrated in most of the examples,
 are an immediate byproduct of this semantic technology, and that allows the results hereby obtained to extend in a non-trivial way the preliminary results in~\cite{ccal:jmar:smar:merging}. 

A more comprehensive understanding of fibred logics, even beyond the disjoint case, is an obvious avenue for future research. But several other narrower alleys have been opened by the work reported in this paper. For a start, despite having done so for all the examples analyzed, we have not been able to obtain a general categorization of the cases when the logic combining two fragments of classical logic fails to be finitely-valued yet still happens to be finitely-\nvalued{}. It seems that a deeper understanding of finite-\nvalued{}ness is still lacking, parallel to the results of~\cite{finval} with respect to finite-valuedness. We have also not managed to prove in a systematic way the completeness of the calculi obtained by the addition of new interaction rules directly from the \nmatrices{} characterizing the fibring of the underlying fragments of classical logic (note that the resulting calculi are known to be complete, as a result of the techniques introduced by Rautenberg in his notable paper~\cite{Rautenberg1981}). We left these completeness proofs open in a few of the examples, as the notion of a valuation respecting an inference rule turns out to be less innocent than it might seem. In order to systematically tackle this problem, it seems that one should try to employ the technique of `rexpansions', from~\cite{rexpansions}, which advocates first expanding the \nmatrix{} at hand in order to be able to split conflicting behaviours in the evaluation of connectives, that may then be simply refined (purged from an undesired value) when one needs to impose an additional rule on it. The completeness proofs we included in our examples are basic instances of the rexpansion technique. Another, more general but related, path to pursue is targeting a deeper understanding of the algebraic properties of \nmatrices{}. A good example of the perplexities brought about by such a seemingly innocent extension of the notion of logical matrix concerns the definition of derived connectives by abbreviation, which amounts to a straightforward matter for operations on matrices but which brings unsuspected difficulties in \nmatrices{}.%
\footnote{The problem is that, while it is true that any logical matrix~$\Mt$ and any translation~$\bt:\Xi\longrightarrow  L_\Sigma(P)$ allow one to interpret unambiguously any derived connective, so that $\Val_P(\Mt^\bt)=\{v\circ\bt:v\in\Val_P(\Mt)\}$, the latter equality is not true, in general, if $\Mt$ is an \nmatrix{}.  For an example of that, suppose the $1$-place derived connective ${\sim}\in\Xi$ is introduced through~$\bt$ as $\lambda p_1.\,p_1\impscript{}{}\botop$, where~$\impscript{}{}$ is the Boolean implication and~$\botop$ is an unrestrained $0$-place connective. In that case, the induced \nmatrix{} $\Mt^\bt$ would take ${\sim}_{\Mt^\bt}=\{0,1\}$, while the original \nmatrix{}~$\Mt$ can only allow ${\sim}_{\Mt}$ to be affirmation connective or the negation connective.}
In particular, a better understanding of more general applicability conditions for our Prop.~\ref{liftspecial} can very well depend on such a fundamental study of \nmatrices{}. Finally, it is important to get a better grip on the role of saturation in the process of fibring logics, and its interplay with strict products of \nmatrices{}. As we have seen, it is sometimes sufficient to require $k$-saturation for finite~$k$; we believe though that other milder forms of saturation may play a key role in obtaining simpler (in particular, denumerable) semantics for combined logics.

In closing, it is worth noting that the essential role played by the saturation requirement in order to explain the semantics of the combined logics seems to suggest that the emergence of interaction principles is connected with a lack of expressiveness of the 
standard Tarskian framework for the study of logics (as hinted also in~\cite{metafib}),
and that 
the outcome of the present investigation would be entirely different
if we were to adopt multiple-conclusion logics, after~\cite{SS:MCL:78}.

\bibliographystyle{plain}  

 \end{document}